\documentclass[twoside,leqno,twocolumn]{article}

\usepackage{graphicx}
\usepackage{amsmath}
\usepackage{amssymb}
\usepackage{mathptmx}
\usepackage{xspace}
\usepackage[ruled,algo2e,noline,noend]{algorithm2e}
\usepackage{cite}
\usepackage{wrapfig}
\usepackage[font=small,labelfont=bf]{caption}
\usepackage[font=small]{subcaption}
\usepackage{thm-restate}
\usepackage{enumitem}

\usepackage{flushend}

\newcommand{\BO}[1]{{O}\left(#1\right)}

\newcommand{\BT}[1]{{\Theta}\left(#1\right)}

\newcommand{\m}{M$_3$\xspace}
\newcommand{\sn}{\sqrt{n}}
\newcommand{\sm}{\sqrt{m}}
\newcommand{\snm}{\sqrt{n/m}}
\renewcommand{\l}{\langle}
\renewcommand{\r}{\rangle}
\newcommand{\pair}[1]{\left\langle #1 \right\rangle}

\SetKwProg{Map}{Map:}{ }{}
\SetKwProg{Reduce}{Reduce:}{ }{}
\SetKw{Emit}{emit}
\SetKw{To}{to}
\newcommand{\Let}[2]{#1 $\leftarrow$ #2}

\usepackage{ltexpprt}

\renewenvironment{proof}{\emph{Proof.}}{\hfill$\square$}

\title{Experimental Evaluation of Multi-Round Matrix Multiplication on
MapReduce}
\author{Matteo Ceccarello \and Francesco Silvestri}
\date{University of Padova, Dip. Ingegneria dell'Informazione, Padova,
Italy,  \texttt{\{ceccarel,silvest1\}@dei.unipd.it}}
\begin{document}

\maketitle

\begin{abstract}
  A common approach in the design of MapReduce algorithms is to
  minimize the number of rounds. Indeed, there are many examples in
  the literature of monolithic MapReduce algorithms, which are
  algorithms requiring just one or two rounds. However, we claim that
  the design of monolithic algorithms may not be the best approach in
  cloud systems.  Indeed, multi-round algorithms may exploit some
  features of cloud platforms by suitably setting the round number
  according to the execution context.

  In this paper we carry out an experimental study of multi-round
  MapReduce algorithms aiming at investigating the performance of the
  multi-round approach.  We use matrix multiplication as a case
  study. We first propose a scalable Hadoop library, named \emph{\m},
  for matrix multiplication in the dense and sparse cases which allows
  to tradeoff round number with the amount of data shuffled in each
  round and the amount of memory required by reduce functions. Then,
  we present an extensive study of this library on an in-house cluster
  and on Amazon Web Services aiming at showing its performance and at
  comparing monolithic and multi-round approaches.  The experiments
  show that, even without a low level optimization, it is possible to
  design multi-round algorithms with a small running time overhead.

\end{abstract}


\section{Introduction}\label{introduction}
MapReduce is a computational paradigm\footnote{In the paper, the term
  MapReduce denotes the abstract programming model and {not} the
  particular implementation developed by Google.} for processing
large-scale data sets in a sequence of rounds executed on
conglomerates of commodity servers~\cite{DeanG08}.  This paradigm, and
in particular its open source implementation Hadoop~\cite{Hadoop12},
has emerged as a de facto standard and has been widely adopted by a
number of large Web companies (e.g., Google, Yahoo!, Amazon,
Microsoft) and universities (e.g., CMU, Cornell)~\cite{KarloffSV10}.
MapReduce was initially introduced for log processing and web indexing
but it has also been successfully used for other applications,
including machine learning~\cite{Mahout11}, data
mining~\cite{AfratiSSU13}, scientific computing~\cite{SriramaJV12},
and bioinformatics~\cite{Taylor12}.  Informally, a MapReduce algorithm
transforms an input multiset of key-value pairs into an output
multiset of key-value pairs in a sequence of \emph{rounds}.  Each
round consists of three steps: each input pair is individually
transformed into a multiset of new pairs by a map function (\emph{map
  step}); then the new pairs are grouped by key (\emph{shuffle step});
finally, each group of pairs with the same key is processed,
separately for each key, by a reduce function that produces the next
new set of key-value pairs (\emph{reduce step}).
The MapReduce paradigm has a functional flavor, in that it merely
requires that the algorithm designer specifies the computation in
terms of map and reduce functions.  This enables parallelism without
forcing an algorithm to cater for the explicit allocation of
processing resources. Nevertheless, the paradigm implicitly posits the
existence of an underlying unstructured and possibly heterogeneous
parallel infrastructure, where the computation is eventually run.  For
these reasons, the MapReduce paradigm has gained popularity in recent
years in cloud services for the development of large scale
computations.  Indeed, MapReduce is currently offered as a service by
prominent cloud providers, such as Amazon Web Services and Microsoft
Azure.

As MapReduce is increasingly used for solving computational hungry
problems on large datasets, it is crucial to design efficient and
scalable algorithms.  Many research efforts have been dedicated to
capture efficiency in MapReduce algorithms
(e.g.,\cite{LinS10,KarloffSV10,PietracaprinaPRSU12,AfratiSSU13,GoodrichSZ11}).
The major source of inefficiency is the communication required for
moving data from mappers to reducers in the shuffle step: since
communication is a major factor determining the performance of
algorithms on current computing systems, the amount of shuffled data
should be minimized.  Another issue that is usually taken into account
for improving performance is the number of rounds of a MapReduce
algorithm, since the initial setups of a round and the shuffle step
are very costly operations.  Thus, several MapReduce algorithms have
been designed that require a very small number, usually one or two, of
rounds (e.g.,~\cite{AfratiSSU13,KarloffSV10}).  We denote with
\emph{monolithic algorithm} a MapReduce algorithm requiring one or two
rounds.  A common approach for obtaining a monolithic algorithm is to
decompose the problem into small subproblems which are executed
concurrently in a single round, with each subproblem solved by a
single application of the reduce function.

In this paper, we claim that the design of monolithic algorithms may
not be the best approach for long running MapReduce computations in
cloud systems.  Although it is true that the best performance is
usually reached by reducing the round number, monolithic MapReduce
algorithms may not exploit some features of cloud computing. Some
examples follow.

\begin{itemize}[itemsep=-1ex,topsep=-2ex]
\item \emph{Service market.} Some cloud providers offer a market where
  it is possible to bid on the cost of a service and to use it only
  when the actual price is below the bid.  For reducing computing
  costs of long but low-priority computations, it would be desirable
  to develop MapReduce algorithms that can be stopped and restarted
  according to the price of the service.  Unfortunately, current
  implementations of the MapReduce paradigm do not allow to stop a
  computation at an arbitrary point and then restart it.  However, it
  is possible, like in Hadoop, to restart a computation from the
  beginning of the round that has been interrupted, losing the work
  that was already executed in that round.  This clearly penalizes
  monolithic algorithms that consist of just a few long rounds:
  stopping a round and waiting for a better price might not be
  convenient since a significant part of previous work has to be
  executed again.  It would be then interesting to develop MapReduce
  algorithms featuring a larger number of short rounds in order to
  reduce the amount of discarded work.
 
\item \emph{Resource requirements.} Some MapReduce algorithms need
  strong resource requirements in order to be monolithic.  For
  instance, the multiplication in a semiring of two dense matrices of
  size $\sn\times \sn$ in two rounds
  \emph{must}~\cite{PietracaprinaPRSU12} exchange approximately
  $n\snm$ pairs during each shuffle, where $m$ denotes an upper bound
  to the memory that can be used by a map/reduce function. This amount
  of data is linear in the input size only if $m\sim n$, that is, when
  the multiplication can be almost solved sequentially. Similar
  requirements are also needed for joining
  relations~\cite{AfratiSSU13} and enumerating triangles in a
  graph~\cite{ParkSKP14}.  In a big-data era, the required local
  memory may exceed system resources and the huge amount of data in
  the shuffle step could arise issues related to network performance
  and system failure.  Indeed, the network may be subject to
  congestion due to the large amount of data created within a small
  time interval, penalizing the scalability and fault tolerance of the
  network.  It is then desirable to design algorithms that can
  tradeoff round number and resource requirements.  We also observe
  that distributing a large computation among different rounds may
  help to checkpoint the computation and thus to restore it if the
  system completely fails or
  is heavily penalized by a fault.\\
\end{itemize}

It is then interesting to investigate multi-round algorithms where the
round number can be set according to the execution context.  However
it is challenging to guarantee performance similar to the ones
provided by monolithic algorithms.  The first issue is that, by
distributing the computation among different rounds, the total amount
of communication in the shuffle steps should not increase with respect
to the monolithic version since, as we have already mentioned,
communication is the main bottleneck.  Moreover, the latency required
by each new round should be amortized by a sufficient amount of
computation and communication performed within the round.

\paragraph{Our results.}
In this paper we carry out an experimental study of multi-round
MapReduce algorithms for matrix multiplication aiming at investigating
the performance of the multi-round approach.  Matrix multiplication is
an important building block for many problems arising in different
contexts, in particular scientific computing and graph
processing~\cite{GolubL12}.  
Furthermore, we believe that matrix
multiplication is an interesting problem for studying multi-round
algorithms since its high, but still tractable, communication and
computation requirements allow to significantly load the system and to
assess the performance under stress.
The results provided in the paper are the following:

\begin{enumerate}
 \item We propose a scalable Hadoop library, named \emph{\m} (\textbf{M}atrix
\textbf{M}ultiplication in \textbf{M}apReduce), for performing dense and sparse
matrix multiplication in Hadoop.  The library is based on the multi-round
algorithms proposed in~\cite{PietracaprinaPRSU12} which exploit a 3D
decomposition of the problem.  We fill the gap between the
theoretical and high-level results in~\cite{PietracaprinaPRSU12} and the actual
implementation by providing a detailed description of the map and reduce
functions.
For the sake of
completeness, \m also contains a MapReduce algorithm for
dense matrix multiplication based on a 2D approach, which is described in 
Section~\ref{sec:2D}. The algorithms exhibit a tradeoff among round number, the
amount of data in the shuffle steps and the amount of memory required by each
reduce function. The library is publicly available at
\texttt{http://www.dei.unipd.it/m3}.
 
\item We carry out an extensive experimental evaluation of the \m
  library on an in-house cluster and on Amazon Elastic MapReduce, and
  show that the 3D algorithms efficiently scale with increasing input
  size and processor number.  We compare the performance of the
  multi-round approach with the monolithic one by measuring the
  performance of our algorithms when the round number increases.  The
  experiments show that the running times are mainly dominated by the
  amount of communication, while the round number has a limited impact
  on performance. This fact suggests that the common use (in
  particular in more application-oriented contexts) to only focus on
  round number when designing MapReduce algorithms is not a best
  practice for improving performance. The results also give evidence
  that multi-round algorithms have performance comparable with
  monolithic ones (assuming a similar total amount of communication)
  even on the Hadoop framework, which is not the most suitable
  MapReduce implementation for executing multi-round
  algorithms. Indeed, the architectural choice of Hadoop to store
  pairs between rounds on the distributed file system HDFS
  significantly penalizes the performance. We believe that other
  implementations (e.g., Spark~\cite{Zaharia10}) may almost remove the performance gap
  and stimulate the adoption of
  multi-round algorithms.\\
\end{enumerate}

We remark that the performance
of our algorithms should not be compared with state-of-the-art
research on parallel sparse/dense matrix multiplication in High
Performance Computing (e.g.,~\cite{BulucG12}).  Indeed, we are
interested in investigating algorithm performance in cloud systems
with MapReduce. In these settings, abstraction layers add a
significant burden to algorithm performance. Research on the reduction
of this load is crucial (see, e.g.~\cite{Zaharia10,Plimpton11}), but
it is out of the scope of this paper.

\paragraph{Previous Work.}
The MapReduce~\cite{DeanG08} paradigm has been widely studied and we
refer to the book by Lin and Dyer~\cite{LinD10} for a survey.  The
design of efficient MapReduce algorithms has been investigated from
practical and theoretical perspectives.  For instance, best practices
in designing large-scale algorithms in MapReduce are proposed
in~\cite{Afrati12,LinS10}, while theoretical models for analyzing
MapReduce are studied
in~\cite{KarloffSV10,PietracaprinaPRSU12,GoodrichSZ11,AfratiSSU13}.
The majority of the algorithms requires just a few rounds, although
monolithic solutions are not known for some problems (e.g., computing
the diameter of a graph~\cite{CeccarelloPPU14} or matrix
inversion~\cite{PietracaprinaPRSU12}).
Multi-round algorithms that trade round number with resource
requirements have been proposed, with a theoretical approach,
in~\cite{PietracaprinaPRSU12} for some linear algebra problems,
including matrix multiplication, in~\cite{GoodrichSZ11} for sorting
and searching, and in~\cite{ParkSKP14} for triangle enumeration. To
the best of our knowledge, the only experimental analysis of
multi-round algorithms has been provided in~\cite{ParkSKP14} for
triangle enumeration and shows that the performance of a multi-round
approach is equivalent to a monolithic one.

Matrix multiplication is one of the most studied problems in the
literature and we refer to the book by Golub and Van
Loan~\cite{GolubL12} for further insights.  However, to the best of
our knowledge, this problem has been scarcely studied in MapReduce:
ignoring naive implementations available on the web, the only
scientific works are the aforementioned~\cite{PietracaprinaPRSU12},
and~\cite{Seo10}. The latter introduces a library, named HAMA,
including MapReduce algorithms for matrix multiplication and conjugate
gradient for matrices stored in HBase databases.

\paragraph{Comparison with previous work.}
This work completes the theoretical investigations
in~\cite{PietracaprinaPRSU12} by integrating the original description
with more details that allow to fill the gap between theory and
practice, and by proposing an Hadoop implementation.  The \m library
includes as special case the algorithms in HAMA~\cite{Seo10}: the
iterative approach proposed in HAMA requires $\sn$ rounds for
multiplying a matrix of size $\sn\times \sn$ and is a special case of
the algorithm based on the 2D approach described in
Section~\ref{sec:2D} (i.e., it suffices to set $m=\sqrt{n}$ and
$\rho=1$).  Since the HAMA library assumes a very different input
representation and our experiments show that the aforementioned 2D
approach is slower than the 3D approach, we do not take into account
HAMA in this paper. Moreover, it should be said that the algorithms
in~\cite{Seo10} break the functional approach of MapReduce by allowing
concurrent accesses to the distributed file system within each map and
reduce function, and that the experiments have been carried out on
small inputs (input side $\sqrt{n} \le$ 5000, while we run our
algorithms with input side $\sqrt{n} \ge$ 16000).

\paragraph{Paper organization.} 
Section~\ref{sec:preliminaries} introduces the matrix multiplication
problem, the MapReduce framework and the experimental settings.
Section~\ref{sec:algorithms} describes and analyzes the proposed
MapReduce algorithms.  Section~\ref{sec:implementation} shows some
additional technicalities required by the Hadoop implementation.
Section~\ref{sec:experiments} provides and explains the experimental
results.  Finally, Section~\ref{sec:conclusion} gives some final
remarks.


\section{Preliminaries}\label{sec:preliminaries}
\paragraph{Matrix multiplication.}\label{sec:mmult}
The focus of the paper is matrix multiplication in a general semiring,
that is we are ruling out Strassen-like algorithms.  For the sake of
simplicity, we focus on square matrices of size $\sn\times\sn$.  We
denote with $A$ and $B$ the two input matrices and with $C=A\cdot B$
the output matrix.
In the sparse case, we denote with $0\leq \delta\leq 1$ the density of
non-zero entries in a given matrix of size $\sn\times\sn$, that is the
number of non-zero entries is $\delta n$.  A random Erd\"{o}s-R\'enyi
matrix of size $\sn\times\sn$ and parameter $0\leq \delta\leq 1$ is a
matrix where each entry is non-zero with probability $\delta$
independently; the expected density is clearly $\delta$.  When
$\delta<<1/n^{1/4}$, the expected density of the product of two
Erd\"{o}s-R\'enyi matrices is $\delta^2 \sn$~\cite{Ballard13}.

\paragraph{MapReduce and Hadoop.}
A MapReduce algorithm consists of a sequence of \emph{rounds}.  The
input/output of a round is a multiset of key-value pairs $\l k; v \r$,
where $k$ and $v$ denote the key and the value respectively.  The
input of a round may contain the output of previous rounds but also
other pairs.  A round is organized in three steps:
\begin{enumerate}[itemsep=-1ex,topsep=-2ex]
\item \emph{Map step}: each input pair is individually given in input
  to a \emph{map function}, which returns a new multiset of pairs.
  The pairs returned by all applications of the map function are
  called \emph{intermediate pairs}.
\item \emph{Shuffle step}: the intermediate pairs are grouped by key.
\item \emph{Reduce step}: each group of pairs associated with the same
  key is individually given to a \emph{reduce function}, which returns
  a new multiset of pairs.  The pairs returned by all applications of
  the reduce function are the
  \emph{output pairs} of the round.\\
\end{enumerate}
For simplicity we refer to a single application of the map (resp.,
reduce) function with \emph{mapper} (resp., \emph{reducer}).  For a
MapReduce algorithm, we denote with \emph{round number} the number of
rounds, with \emph{shuffle size} the maximum number of intermediate
pairs in each round, and with \emph{reducer size} the maximum amount
of memory words required by each reducer.\footnote{The memory
  requirements of the map functions of our algorithms are very small
  constants, and therefore they are ignored.}

Hadoop\cite{Hadoop12}
is the most used open source implementation of the MapReduce
paradigm. A MapReduce algorithm is executed by Hadoop on $p$
processors roughly as follows.  In each round, each processor is
associated with a fixed number of \emph{map tasks} and \emph{reduce
  tasks}.  In the map step, the runtime system evenly distributes the
input pairs to the map tasks and then each map task applies the map
function to each single pair.  In the reduce step, all groups of
same-key pairs are randomly assigned to each reduce task and then each
reduce task applies the reduce function to each single group.  The
input/output pairs are stored in the distributed file system
HDFS~\cite{Shvachko10}, while intermediate pairs are temporarily
stored on the disk of each machine.  Hadoop allows to personalize
different settings of the framework, in particular it allows to
specify how groups are distributed among reduce tasks by defining a
\emph{partitioner}. A partitioner receives as input a key $k$ and the
total number $T$ of reduce tasks, and it returns a value in $[0,T)$
uniquely denoting the reduce task that will apply the reduce function
to the group associated with key $k$.  

We recall that there exist other efficient implementations of the
MapReduce paradigm, including Spark~\cite{Zaharia10},
MapReduce-MPI\cite{Plimpton11}, Sphere~\cite{GrossmanKDD2008}. Our
algorithms can be executed on these frameworks with minor programming
changes.

\paragraph{Experimental Framework.}
The in-house cluster consists of 16 nodes, each equipped with a 4-core
Intel i7 processor Nehalem @ 3.07GHz, 12 GB of memory, 6 disks of 1TB
and 7200 RPM in RAID0.  The interconnection system is a 10 gigabit
Ethernet network.  The operating system is Debian/Squeeze with kernel
2.6.32.  All experiments have been executed with Hadoop 2.4.0 on 16
nodes (1 master + 16 slaves; a node was both master and slave).

The experiments on Amazon Web Services have been executed on instances
of type \emph{c3.8xlarge} and \emph{i2.xlarge} in the US East
Region.\footnote{{\tt http://aws.amazon.com/ec2/instance-types}.}  The
c3.8xlarge is a compute optimized instance featuring 32 virtual cores
on a physical processor Intel Xeon E5-2680 @2.8GHz, 64 GB of memory, 2
solid state disks of 320 GB; instances are connected by a 10 gigabit
network.  The i2.xlarge is a storage optimized instance featuring 4
virtual cores on a physical processor Intel Xeon E5-2670 @2.5GHz, 32
GB of memory, 1 solid state disk of 800 GB optimized for very high
random I/O performance; instances are connected by a
moderate-performance network.  All experiments have been run on
clusters with 9 (1 master + 8 slave) instances of the same type, with
a Debian/Squeeze operating system and Hadoop 2.2.0.


\section{Algorithms}\label{sec:algorithms}
The \m library contains two MapReduce algorithms for
performing matrix multiplication which are based on the 3D approach
proposed in~\cite{PietracaprinaPRSU12}.  In this section, we first
propose the algorithm for multiplying two square dense matrices. Then, 
we give the algorithm that multiplies two sparse
matrices; although it
applies to any sparse input, we analyze it in the case of random
Erd\"{o}s-R\'enyi matrices.  The algorithms exhibit tradeoff among
round number, shuffle size and reducer size, and rely on the well-know
3D decomposition of the lattice representing the $n^{3/2}$ elementary
products (some of which are zeros in the sparse case) into cubes of a
given size.  In other words, the problem is decomposed into
multiplications of smaller square matrices that are solved
sequentially by the reducers. The subproblems are evenly distributed
among the rounds in order to guarantee the specified tradeoff.

For the sake of comparison, we also provide in Section~\ref{sec:2D} an
algorithm for dense matrix multiplication which is based on the 2D
decomposition of the lattice: the problem is decomposed into the
multiplication of smaller rectangular matrices where the longest side
is $\sn$.  This approach is quite common in naive implementations of
matrix multiplication in MapReduce available on the web: however, as
we will see, this approach is inefficient.  The algorithm exhibits a
tradeoff among round number, shuffle size, and reducer size. This
result does not appear in~\cite{PietracaprinaPRSU12}.

The tradeoffs in our algorithms are highlighted by expressing the
round number as a function of two parameters: the \emph{replication
  factor} $\rho$ and the \emph{subproblem size} $m$.  The replication
factor gives an estimate of the desired volume of intermediate data in
each round, while the subproblem size bounds the memory requirements
of each reducer: indeed, our algorithms guarantee that the shuffle
size is $\BT{\rho}$ times the input size and the reducer size is
$\BT{m}$.

\subsection{3D Algorithm  for Dense Matrix
Multiplication.}\label{sec:3DdenseXXXX}

\begin{algorithm2e}[t]
  \caption{Map and reduce functions of the $r$-th round of the 3D
    dense algorithm, with $0\leq r < \sqrt{n}/(\rho\sm)+1$.}
  \label{algo:map3Ddense}
  
  \DontPrintSemicolon
  \Map{input $\pair{(i,\ell,j); D}$;
    $D$ is $A_{i,j}$ or $B_{i,j}$ if $\ell=-1$,\\$~\qquad$ or
      $C^\ell_{i,j}$ otherwise;}{
    \Switch{$D$}{
      \Case{$D$ is $A_{i.j}$}{
        \For{\Let{$\ell$}{0} \To $\rho -1$}{
          \Let{$h$}{$j - i - \ell$}\;
          \Emit $\pair{(i,j,h); D}$\;
        }
      }
      \Case{$D$ is $B_{i.j}$}{
        \For{\Let{$\ell$}{0} \To $\rho -1$}{
          \Let{$h$}{$i - j - \ell$}\;
          \Emit $\pair{(h,i,j); D}$\;
        }
      }
      \Case{$D$ is $C^\ell_{i.j}$}{
        \If{$r$ is the last round}{
          \Emit $\l (i,-1,j); D \r$\;
        }
        \Else{
          \Let{$h$}{$i + j + \ell + r\rho$}\;
          \Emit $\l (i,h,j); D \r$\;
        }
      }
    }
  }
  \;
  \Reduce{input $\l (i,h,j); \{A_{i,h}, B_{h,j}, C_{i,j}^\ell\}\r$ if
    $0\leq i,j < \snm$ and $h=(i+j+\ell+r\rho)\mod \snm$, or $\l
      (i,-1,j); \{C_{i,j}^0, \ldots C_{i,j}^{\rho-1}\}\r$}{
    \If{$r$ is the last round}{
      \Emit $\l (i,j); \sum_{\ell=0}^{\rho - 1} C_{i,j}^\ell \r$\;
    }
    \Else{
      \Emit $\l (i, \ell, j); C^\ell_{i,j} + A_{i,h}B_{h,j} \r$\;
    }
  }
\end{algorithm2e}

The 3D MapReduce algorithm for dense matrix multiplication was
initially proposed in~\cite{PietracaprinaPRSU12}.  The paper provides
an high level description of the algorithm in the MR computational
model, however it does not provide important details such as a
description of the map and reduce functions that are required for
actually implementing the algorithm in MapReduce.  In this section, we
close this gap by providing a detailed description of the algorithm,
including a pseudocode for the map and reduce functions in
Algorithm~\ref{algo:map3Ddense}.

For any $1\leq m\leq n$ and $1\leq \rho\leq \snm$, the algorithm
requires $R=\sqrt{n}/(\rho\sm)+1$ rounds, the shuffle size is $3\rho
n$, and the reducer size is $3m$. For the sake of simplicity, we
assume that $\sn$ and $\sm$ are integers and that $\sm$ divides $\sn$.
The input matrices $A$, $B$ and the output matrix $C$ are divided into
submatrices of size $\sm\times \sm$, and we denote these submatrices
with $A_{i,j}$, $B_{i,j}$ and $C_{i,j}$, respectively, for $0\leq i,j
< \sqrt{n/m}$.  Clearly, we have $C_{i,j}=\sum_{h=0}^{\sqrt{n/m}-1}
A_{i,h}\cdot B_{h,j}$.  The input matrix $A$ is stored as a collection
of pairs $\l (i,-1,j); A_{i,j} \r$, where $-1$ is used as dummy
value. Matrices $B$ and $C$ are stored similarly.  The division
of the input matrices implies $(n/m)^{3/2}$ products between
submatrices, which are partitioned into $\snm$ groups as follows:
group $G_\ell$, with $0\leq \ell < \snm$, contains the products
$A_{i,h}\cdot B_{h,j}$ for $h=(i+j+\ell)\mod \snm$ and for every
$0\leq i,j <\snm$. We note that each submatrix of $A$ and $B$ appears
exactly once in each group.

The algorithm works in $R$ rounds. In the $r$-th round, with $0\leq r
< R-1$, the algorithm computes all products in $G_\ell$ and the sum
$C_{i,j}^\ell = \sum_{k=0}^{r-1} A_{i, i+j + \ell + k\rho} B_{i+j +
  \ell + k\rho, j}$ for every $r \rho \leq \ell < (r+1)\rho$ and
$0\leq i,j<\snm$; each product $A_{i,h}\cdot B_{h,j}$ is computed
within the reducer associated with key $(i,h,j)$. Eventually, in the
last round $r=R-1$, the matrix $C$ is created by computing $C_{i,j} =
\sum_{\ell=0}^{\rho -1} C_{i,j}^\ell$ for every $0\leq i,j <\snm$.  A
more detailed explanation follows.  At the beginning of the $r$-th
round, the input pairs are $\l (i,-1,j); A_{i,j} \r$, $\l (i,-1,j);
B_{i,j} \r$.  Starting from the second round (i.e., $r\geq 1$) the
algorithm also receives as input the pairs $\l (i,\ell,j);
C^\ell_{i,j} \r$, for every $0\leq i,j < \snm$ and $0\leq \ell <
\rho$, that have been outputted by the previous round.  In the map
step, each input pair is replicated $\rho$ times since it is required
in $\rho$ reducers.  Specifically, for each pair $\l (i,-1,j); A_{i,j}
\r$, the map function emits the pairs $\l (i,j,h); A_{i,j} \r$ with
$h=j-i-\ell \mod \snm$ for every $0\leq \ell < \rho$.  Similarly, for
each pair $\l (i,-1,j); B_{i,j} \r$, the map function emits the pairs
$\l (h,i,j); B_{i,j} \r$ with $h=i-j-\ell \mod \snm$ for every $0\leq
\ell < \rho$.  On the other hand, the map function emits for each pair
$\l (i,\ell,j); C^\ell_{i,j} \r$ the pair $\l (i,h,j); C^\ell_{i,j}
\r$ with $h = i+j + \ell +r\rho \mod \snm$.  Finally, each reducer
associated with key $(i,h,j)$ with $h=i+j+\ell$, for any $0\leq i,j <
\snm$ and $0\leq \ell < \rho$, receives as input $A_{i,h}$, $B_{h,j}$
and $C^\ell_{i,j}$, computes
$C^\ell_{i,j}=C^\ell_{i,j}+A_{i,h}B_{h,j}$ and emits $\l(i,\ell,j),
C_{i,j}^\ell\r$. 

\begin{theorem}
\label{thm:3d}
  The above algorithm requires $R=\sqrt{n}/(\rho\sm)+1$ rounds, the
  shuffle size is $3\rho n$, and the reducer size is $3m$.
\end{theorem}
\begin{proof}
  The correctness of the algorithm is proved
  in~\cite{PietracaprinaPRSU12}.  However, we have to show that each
  reducer receives the correct set of submatrices since pair
  distribution is not taken into account in the original paper.  Each
  product $A_{i,k}\cdot B_{k,j}$ is computed when $\ell=(k-i-j)\mod
  \rho$ and in round $r= (k-i-j-\ell)/\rho$.  In this round, the
  mapper associated with pair $\l(i,k); A_{i,k}\r$ emits the pair
  $\l(i,k, k-i-\ell -r\rho); A_{i,k}\r$ where $k-i-\ell -r\rho=j$ for
  the above values of $\ell$ and $r$.  A similarly argument applies to
  $B_{h,j}$ and $C^\ell_{i,j}$.

  Since each submatrix of $A$ and $B$ is replicated $\rho$ times and
  there are $\rho n/m$ matrices $C_{i,j}^\ell$, the shuffle size is
  $3\rho n$.  Each reducer requires at most $3m$ memory words for
  computing $C_{i,j}^\ell=C_{i,j}^\ell+A_{i,i+j+\ell+r\rho}\cdot
  B_{i+j+\ell+r\rho,j}$.
\end{proof}

We observe that there are $3n \snm$ pairs in each shuffle in a
two-round algorithm for matrix multiplication (i.e, $\rho=\snm$). As
shown in~\cite{PietracaprinaPRSU12}, this is the best we can get if
only semiring operations are allowed.  Finally, we note that the
algorithm guarantees that almost all mappers perform the same amount
of work because each submatrix of $A$ and $B$ is replicated $\rho$
times.

\subsection{3D Algorithm  for Sparse Matrix
Multiplication.}\label{sec:3DsparseXXX}
\newcommand{\smp}{\sqrt{m'}} An algorithm for sparse matrix
multiplication easily follows from the previous 3D algorithm by
exploiting the sparsity of the input and output matrices.  We assume
that the input matrices are random Erd\"{o}s-R\'enyi matrices with
size $\sn\times \sn$ and expected density $\delta<<n^{1/4}$; we will
later see how to extend the result to the general case.  We recall
that the output density of the product of two random Erd\"{o}s-R\'enyi
matrices is $\delta_O=\delta^2\sn$~\cite{Ballard13}.

The input matrices $A$ and $B$ and the output matrix $C$ are
partitioned into blocks of size $\smp\times \smp$, where
$m'=m/\delta_O=m/(\delta^2 \sn)$.  Each submatrix is represented as a
list of non-zero entries, although other formats can be used like the
Compressed Row Storage.  The sparse algorithm follows by applying the previous
3D dense algorithm to the submatrices of
size $\smp\times \smp$.  The sparse algorithm exploits the fact that
each submatrix $A_{i,j}$ or $B_{i,j}$ (resp., $C_{i,j}$) is a random
Erd\"{o}s-R\'enyi matrix with expected density $\delta$ (resp.,
$\delta_O)$. Then the expected space required for computing
$A_{i,h}\cdot B_{h,j}$ is $(2\delta + \delta^2 \sn) m' \sim 3m$, and
thus the reducer size is $3m$ even for this algorithm.  The pseudocode
can be easily derived from the pseudocode for the dense case, proposed in
Algorithm~\ref{algo:map3Ddense}, by replacing $m$ with $m'$. 


\begin{theorem}
  The above algorithm requires $R=\delta n^{3/4}/(\rho \sm)+1$ rounds,
  the expected shuffle size is $3\rho \delta^2 n^{3/2}$, and the
  expected reducer size is $3m$.
\end{theorem}
\begin{proof}
  The claim easily follows from Theorem~\ref{thm:3d} and by the
  sparsity of the input/output matrices.
\end{proof}

Suppose that $A$ and $B$ are general sparse matrices with density
$\delta<<1$ and that an approximation of the density size of the
output matrix is known $\tilde \delta_O$. Let $\delta_M=\max\{\delta,
\tilde \delta_0\}$, then it is sufficient to split each input/output
matrix into blocks of size $\smp\times \smp$ with $m'=m/\delta_M$ and
apply the previous sparse algorithm. For improving the load balancing
among reducers, columns and rows of the input matrices should be
randomly permuted~\cite{PietracaprinaPRSU12}.  In this case it is easy
to show that the algorithm requires $R=\sqrt{\delta_M n} /(\rho\sm)+1$
rounds, the expected shuffle size is $3\rho \delta_M n$, and the
expected reducer size is $3m$.  A good approximation of the output
matrices can be computed with a scan of the input matrices (see,
e.g.,~\cite{PaghS14}).  Finally, we observe that if the output size is
not known, then it suffices to set $m'=m/\delta$ but the bounds on
shuffle and reducer sizes do not apply anymore.

\subsection{2D Algorithm for Dense Matrix Multiplication}
\label{sec:2D}

The 2D MapReduce algorithm divides the input matrix $A$ into $n/m$
submatrices of size $m/\sn \times \sn$, the input matrix $B$ into
$n/m$ submatrices of size $\sn\times m/\sn$, and the output matrix
into $(n/m)^2$ submatrices of size $m/\sn \times m/\sn$.  We label
each submatrix of $A$ and $B$ with $A_i$ and $B_i$ and each submatrix
of $C$ with $C_{i,j}$, for any $0\leq i,j < n/m$.  We clearly have
$C_{i,j} = A_i\cdot B_j$.  The input matrix $A$ is stored as a
collection of pairs $\l (i,-1); A_{i} \r$, for any $0\leq i< n/m$,
where $-1$ is used as dummy value; similarly for matrix $B$.  The
output matrix is stored as a collection of pairs $\l (i,j); C_{i,j}
\r$, for any $0\leq i,j< n/m$.

For any $\sn\leq m\leq n$ and $1\leq \rho\leq n/m$, the algorithm
requires $R=n/(\rho m)$ rounds, the shuffle size is $2\rho n$, and the
reducer size is $3m$.  For the sake of simplicity, we assume that
$\sn$ is an integer and that $\sn$ divides $m$.  In the $r$-th round,
with $0\leq r < R$, the algorithm computes the products $C_{i,j} =
A_i\cdot B_j$ with $j=i+\ell +r\rho \mod (n/m)$, for any $0\leq i<
n/m$ and $0\leq \ell < \rho$.  Specifically, at the beginning of the
$r$-th round, the input pairs are $\l (i,-1); A_{i} \r$, $\l (-1,j);
B_{j} \r$.  Then, the map function emits for each pair $\l (i,-1);
A_{i} \r$ the pairs $\l (i,j); A_{i} \r$ with $j=(i+\ell +r\rho ) \mod
(n/m)$ for any $0\leq \ell < \rho$; similarly, for each pair $\l
(-1,j); B_{j} \r$, the mapper emits the pairs $\l (i,j); B_{j} \r$
with $i=(j-\ell -r\rho ) \mod (n/m)$ for any $0\leq \ell < \rho$.
Then, each reducer associated with key $(i,j)$ receives $A_i$ and
$B_j$ and emits the pair $\l(i,j); A_i\cdot B_j\r$. The pseudocode of
the map and reduce functions is in Algorithm~\ref{algo:map2Ddense}.

\begin{algorithm2e}[t]
  \caption{Map and reduce functions of the $r$-th round of the 2D
    algorithm, with $0\leq r < n/(\rho m)$.}
  \label{algo:map2Ddense}

  \DontPrintSemicolon
  \Map{input $\pair{(i,j); D}$; $D$ is $A_i$ if $j=-1$ or $B_j$ if $i=-1$;}{
    \Switch{$D$}{
      \Case{$D$ is $A_i$}{
        \For{\Let{$l$}{$0$} \To $\rho-1$}{
          \Let{$j$}{$(i + l + r\rho) \mod (n/m)$}\;
          \Emit $\pair{(i,j); D}$
        }
      }
      \Case{$D$ is $B_i$}{
        \For{\Let{$l$}{$0$} \To $\rho-1$}{
          \Let{$i$}{$(j - l - r\rho) \mod (n/m)$}\;
          \Emit $\pair{(i,j); D}$\;
        }
      }
    }
  }

  \;
    
  \Reduce{input $\pair{(i,j); \{A_i, B_j\}}$ for $0 \le i,j < n/m$}{
    \Emit $\pair{(i,j); A_i \cdot B_j}$\;
  }

\end{algorithm2e}

\begin{theorem}
  The above algorithm requires $R=n/(\rho m)$ rounds, the shuffle size
  is $2\rho n$, and the reducer size is $3m$.
\end{theorem}
\begin{proof}
  Subproblem $(i,j)$ is executed with $\ell=(j-i)\mod \rho$ in round
  $r=(j-i-\ell)/\rho$.  The reducer associated with key $(i,j)$
  correctly receives matrices $A_i$ and $B_j$ in round $r$ since the
  mapper with input $\l(i,j); A_i\r$ emits the pair
  $\l(i,i+\ell+r\rho); A_i\r$ where $i+\ell+r\rho=j$ for the above
  values of $\ell$ and $r$ (a similar argument applies for $B_j$).
  The shuffle size is $2\rho n$ since there are at most $\rho$ copies
  of each submatrix $A_i$ and $B_j$. Each reducer just requires at
  most $3m$ memory words for computing $C_{i,j}=A_i\cdot B_j$.
\end{proof}

Our algorithm guarantees that each mapper performs the same amount of
work: indeed, each submatrix of $A$ and $B$ is replicated exactly
$\rho$ times.  We observe that a naive way to distribute subproblems
among rounds may significantly unbalance the work performed by
mappers: in the worst case, there can be one submatrix of $A$
replicated $n/m$ times, while the remaining submatrices of $A$ are
replicated once (similarly for $B$); then one mapper requires
$\BO{n/m}$ work, while all the remaining mappers perform $\BO{1}$
work.


\section{Implementation}\label{sec:implementation}

The algorithms described in previous sections are implemented as
Hadoop jobs, one for each round. Each job comprises a single map
and a single reduce function.
Matrices are represented as 
{\tt SequenceFile}s where keys are triplets or pairs (for 3D and 2D
multiplication, respectively) and values are serialized objects
representing blocks. The representation used for blocks will vary,
depending on the matrix being dense or sparse.
For dense matrices, each serialized block is the sequence of elements
of the matrix in row-major order. As for sparse matrices, only
non-zero elements are serialized along with their indexes.
The entries of the matrices are doubles. Further details on the Hadoop
configuration are provided in Section~\ref{app:conf}.

\begin{figure}[t]
\centerline{
\includegraphics[width=.5\textwidth]{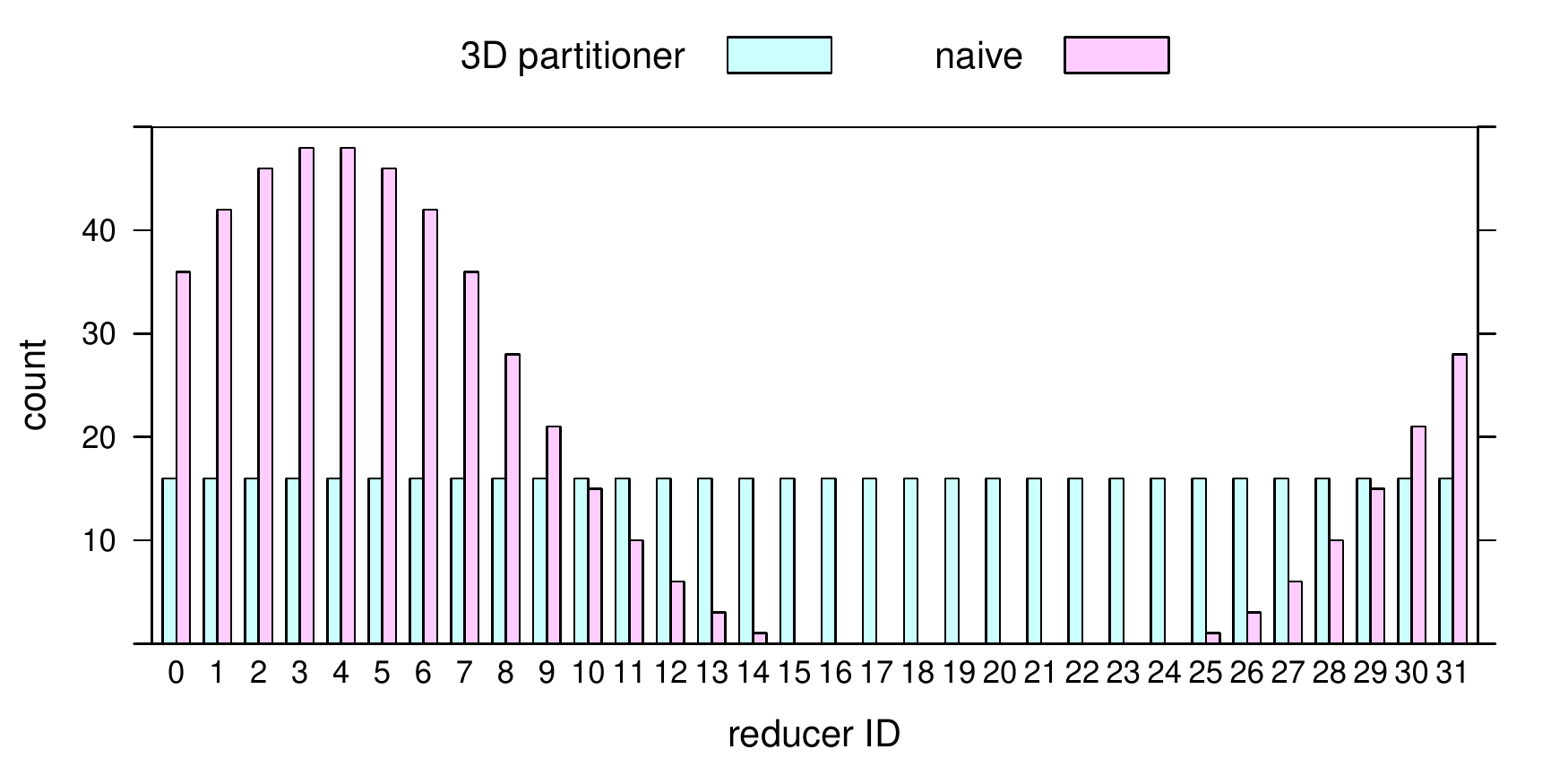}
}
\caption{The figure reports the number of reducers in the $i$-th
  reduce task using a simple partitioner and the one proposed in
  Algorithm~\ref{algo:partitioner} when $\sn=32000$, $\sm=4000$, and
  $\rho=8$ (only first round).}
\label{fig:partitioner}
\end{figure}

\subsection{Map and reduce implementation}
Map functions are straightforward implementations of the algorithms  in
Section~\ref{sec:algorithms}. Reduce functions, while being adherent to
the given specifications, present a few caveats from both
a correctness and a performance perspective.
Reduce functions in Hadoop take input values as instances of the {\tt
  Iterable} interface, giving Hadoop developers the freedom of
changing the actual underlying implementation without breaking client
code. If one needs all the input values at the same time to perform
the reduce computation (as is our case), then a \emph{deep copy} of
each value must be saved in local variables or collections to ensure
correctness. Simply storing the reference does not work, since the
implementation of the {\tt Iterable} interface used by Hadoop returns
a mutable object for each {\tt next} call (at least in the
  most recent version, {\tt 2.4.0}). 
Failing to perform a deep copy of
the input will make the reducer use the same block three times instead
of the proper blocks from matrices $A$, $B$, and $C$.
Note that by performing this copy we incur an unavoidable performance
penalty. 

As for the performance in general, it is crucial to use a fast linear
algebra library for local multiplications.  For the dense case, we use
the {\tt JBLAS} library\footnote{{\tt http://jblas.org/}}, that
provides Java bindings to the fast BLAS/LAPACK
routines. Unfortunately, {\tt JBLAS} does not support sparse matrices,
for which we use {\tt MTJ}\footnote{{\tt
    https://github.com/fommil/matrix-toolkits-java/}}. Although
this is the fastest library for sparse matrix multiplication among the
tested ones, it is orders of magnitude slower than {\tt JBLAS} on
inputs of the same size
thus preventing us to perform complete experiments on sparse matrices.

\begin{algorithm2e}[t]
  \caption{Partitioner of the 3D dense algorithm}
  \label{algo:partitioner}

  \DontPrintSemicolon

  \SetKwProg{Partitioner}{Partitioner:}{ }{}

  \Partitioner{input: $(i,h,j)$, number of reduce tasks $T$;}{
    \Let{$B$}{$\lfloor \rho n/(mT) \rfloor$}\;
    \Let{$k$}{$i\rho n/m + j\rho + (h \mod \rho)$}\;
    \lIf{$z<B\cdot T$}{\Return{$k/B$}}
    \lElse{\Return{Random number in $[0, T)$}}
  }

\end{algorithm2e}

\subsection{Hadoop configuration}\label{app:conf}

As mentioned in Section~\ref{sec:preliminaries}, we used both a
in-house cluster and Amazon Web Services (AWS) to carry out our
experiments. Our cluster has been configured as follows. HDFS
replication has been set to one, that is, redundancy has been
disabled. We found that the default replication, 3, degraded
performance of $\approx 5\%$ while not providing benefits in our
experimental setting. Each machine of the cluster runs two mappers and
two reducers, each with 3GB of memory, in order to allow us to use
matrix blocks with many elements.

As for the configuration of Hadoop instances provided by AWS, we used the
default configuration, customized by Amazon for each machine instance
type.
The rationale is that AWS provides Hadoop as-a-service, enabling
users to concentrate on their application rather than on the
configuration and management of an Hadoop cluster.

\subsection{Partitioner}\label{sec:partitioner}
The keys of the dense 3D algorithms proposed in the
Section~\ref{sec:algorithms} are triplets $(i,j,k)$.  A partitioner
receives as input a key and the number $T$ of reduce tasks in the
system and returns the identifier $t\in[0,T)$ of the reduce tasks that
will be responsible of the key.  When the key is a triplet $(i,j,k)$,
a common partitioner for these keys is $t = (31^2 i+ 31 j + k) \mod
T$.  However, as Figure~\ref{fig:partitioner} shows, this function is
not able to evenly distribute reducers among the $T$ reduce tasks.

We propose an alternative partitioner. In the $r$-th round there are
$\rho(n/m)$ reducers denoted by keys $(i,h,j)$ with $0\leq i,j< \snm$
and $(i+j+r\rho) \mod \snm \leq h < (i+j+(r+1)\rho)\mod \snm$.  The
partitioner uniquely maps each key $(i,h,j)$ in the range $[0,\rho
n/m)$ by setting $z= i \rho n/m + j\rho + h'$ with $h'=h \mod
\rho$. We note that the mapping consists of a row-major ordering on
the first, third and second coordinates, where the second coordinate
has been adjusted to be in the range $[0,\rho)$.  Keys mapped in
$[0,T\lfloor \rho n/(mT)\rfloor)$ are evenly distributed among the
reduce tasks, while the remaining at most $T-1$ keys are randomly
distributed among the tasks.  The pseudocode is given in
Algorithm~\ref{algo:partitioner}.  This partitioner is also used
in the 3D sparse algorithm, while a slightly different approach is used
for the 2D algorithm.


\section{Experiments}\label{sec:experiments}
In this section we propose our experimental analysis of the \m
library.  To better capture the impact of the parameters of the
computational model on the performance, we focus on the dense
algorithm which heavily loads the communication and computational
resources and whose performance does not depend on the particular
input matrix. We aim at answering the following questions: 
\begin{description}[itemsep=-1ex,topsep=-2ex]
 \item[Q1:] How much does the subproblem size $m$ affect the performance?
 \item[Q2:] Are the performance of multi-round algorithms comparable with a
monolithic algorithm?
 \item[Q3:] Which are the major factors affecting the running time?
 \item[Q4:] Does the algorithm scale efficiently with processor
   number?
 \item[Q5:] How much is the performance gap between the 2D and 3D
   approaches?
 \item[Q6:] Does the sparse algorithm efficiently exploit the input sparsity?\\
\end{description}

In Section~\ref{sec:experiments-cluster} we discuss the results on the
in-house cluster. Our claims are also supported by the experiments
carried out in the cloud service Amazon Elastic MapReduce and proposed
in Section~\ref{sec:experiments-amazon}.

\subsection{Experiments on the in-house cluster}
\label{sec:experiments-cluster}

In this section we report our investigations on the in-house cluster
targeting all of the above questions.

\begin{figure}[t]
  \centering
  \includegraphics[width=\columnwidth]{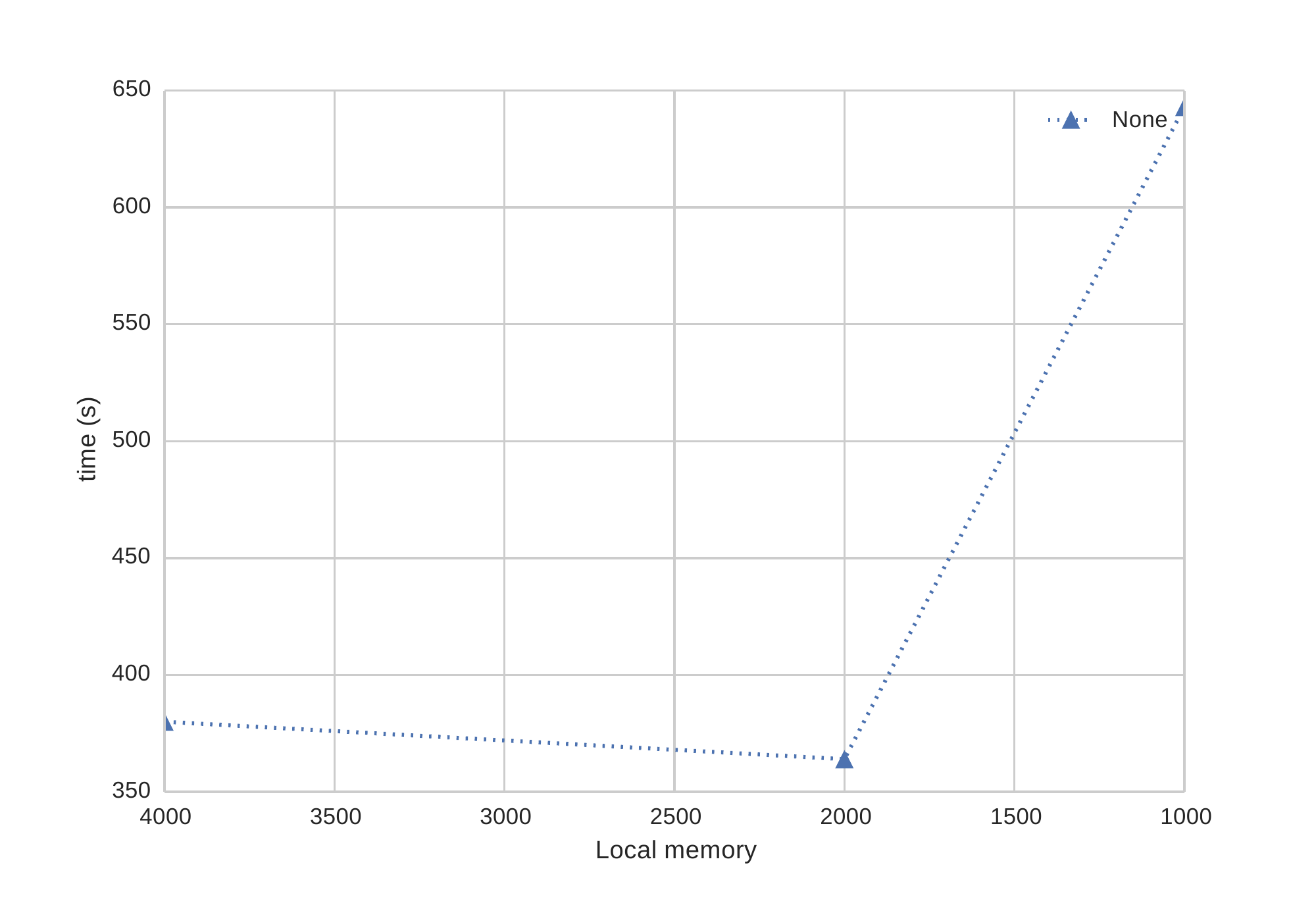}
  \caption{Time vs subproblem size with $\sn\in\{16000,32000\}$. The
    label \emph{max} (resp., \emph{min}) means that it has been used
    maximum (resp., minimum) replication, that is $\rho=\snm$
    (resp., $\rho=1$).}
  \label{fig:local-memory}
\end{figure}

\paragraph{Q1: Impact of subproblem size $m$.}
The total amount of shuffled data in all rounds is $\BO{n\snm}$. This
quantity is independent of the replication factor $\rho$, but
increases as $m$ decreases.  Since the performance of an algorithm
strongly relies on the total amount of shuffled data, it is reasonable
to select a large $m$ and then to decompose the problem into larger
subproblems.  On the other hand, a large value of $m$ reduces
parallelism since there are $(n/m)^{3/2}$ independent subproblems.
Moreover, it increases the load on each machine since a matrix
multiplication of size $\sm\times \sm$ is computed by each reducer
with $\BT{m^{3/2}}$ work and $\BT{m}$ memory, which may exceed
hardware limits.

This fact is highlighted in Figure~\ref{fig:local-memory} which shows
the running time for $\sn\in\{16000, 32000\}$ and
 $\sm\in\{1000, 2000, 4000\}$;
the experiments have been carried out with $\rho=\snm$ and with
$\rho=1$, that is with a monolithic approach and the extreme case of
multi-round.  All experiments show that the performance improves with
larger values of $m$, but the gain decreases for larger values: with
input side $\sqrt{n}=32000$ and $\rho=\snm$, the performance gain when $\sm$
moves from $1000$ to $2000$ is $1.99$, while it is $1.12$ when it
moves from $2000$ to $4000$.  The monolithic approach
is slightly less sensible to variations of $m$ since it can exploit
much more parallelism in each round than a multi-round approach.  The
value $\sm=8000$ is missing in the experiments since all executions
failed due to an out-of-memory error, thus reinforcing the previous claim that
larger values of $m$ may exceed hardware limits. Although each submatrix
requires
 488MB, the actual memory requirements are much more due to the
internal structure of Hadoop, which may keep in the memory of each
processor the submatrices outputted by map tasks as well as the
submatrices required in input by the reduce tasks. Advanced
optimizations can increase the maximum value of $m$, but do not
significantly change the proposed results.

\paragraph{Q2: Impact of replication factor $\rho$.}\label{sec:replfactXXX}

\begin{figure*}
  \centering
  \begin{subfigure}[b]{\columnwidth}
    \includegraphics[width=\textwidth]{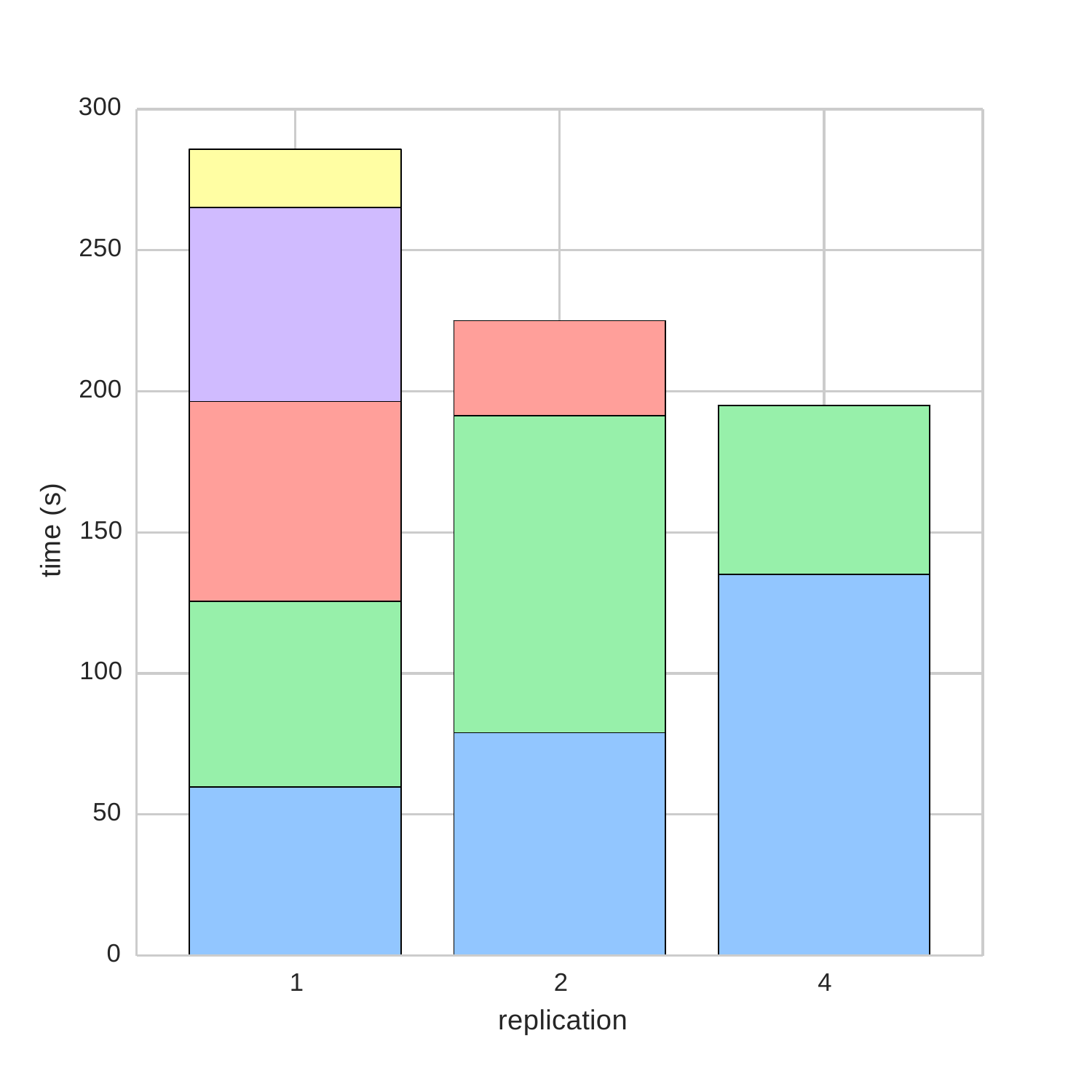}
    \caption{ $\sn=16000$, $\rho=\{1,2,4\}$.}
    \label{fig:replication-plot-16000}
  \end{subfigure}
  \hfil
  \begin{subfigure}[b]{\columnwidth}
    \includegraphics[width=\textwidth]{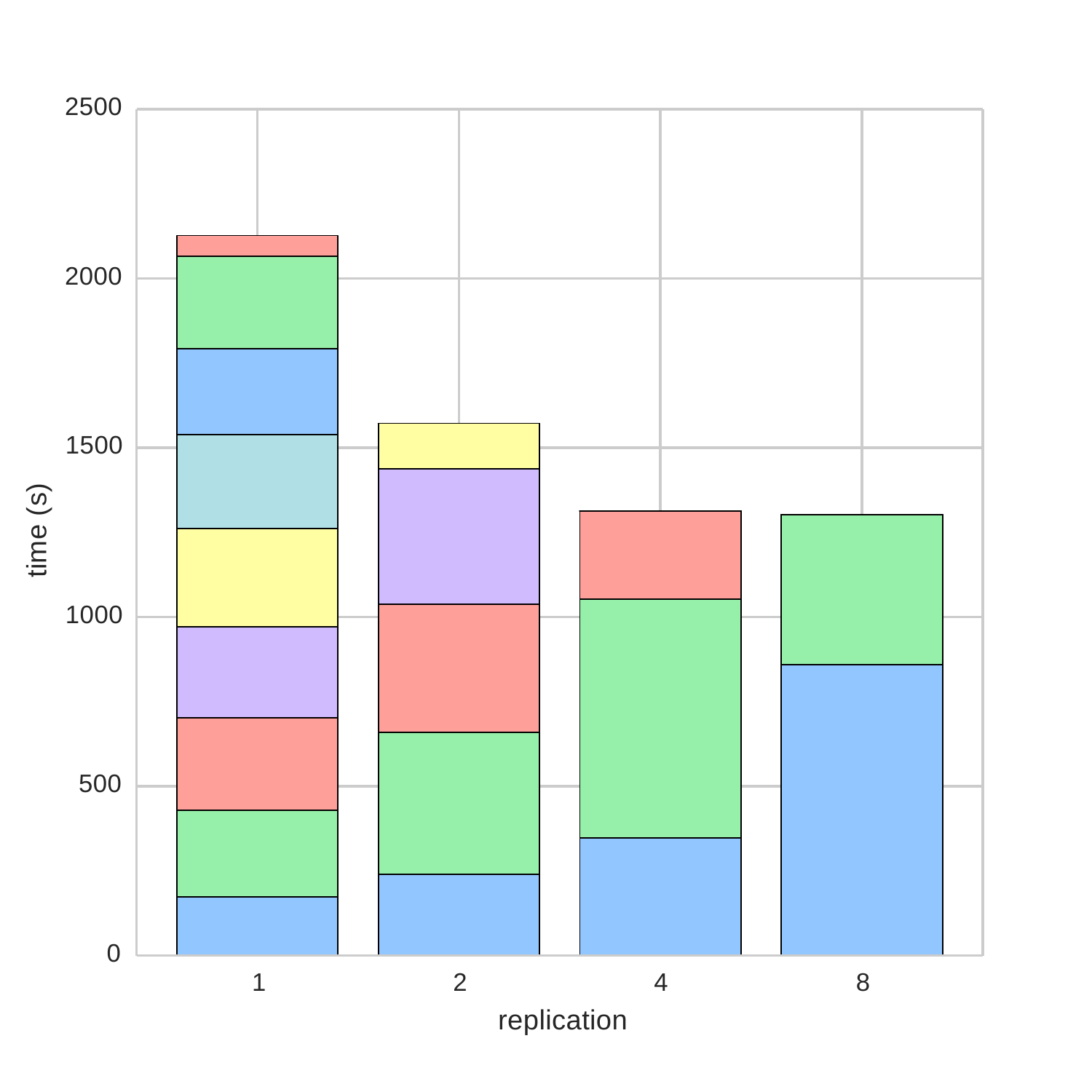}
    \caption{$\sn=32000$, $\rho=\{1,2,4,8\}$.}
    \label{fig:replication-plot-32000}
  \end{subfigure}
  \caption{Time vs replication charts. In each bar of the histogram,
    the $i$-th colored block denotes the time of the $i$-th round.}
\end{figure*}

\begin{figure*}
  \centering
  \begin{subfigure}[b]{\columnwidth}
    \includegraphics[width=\textwidth]{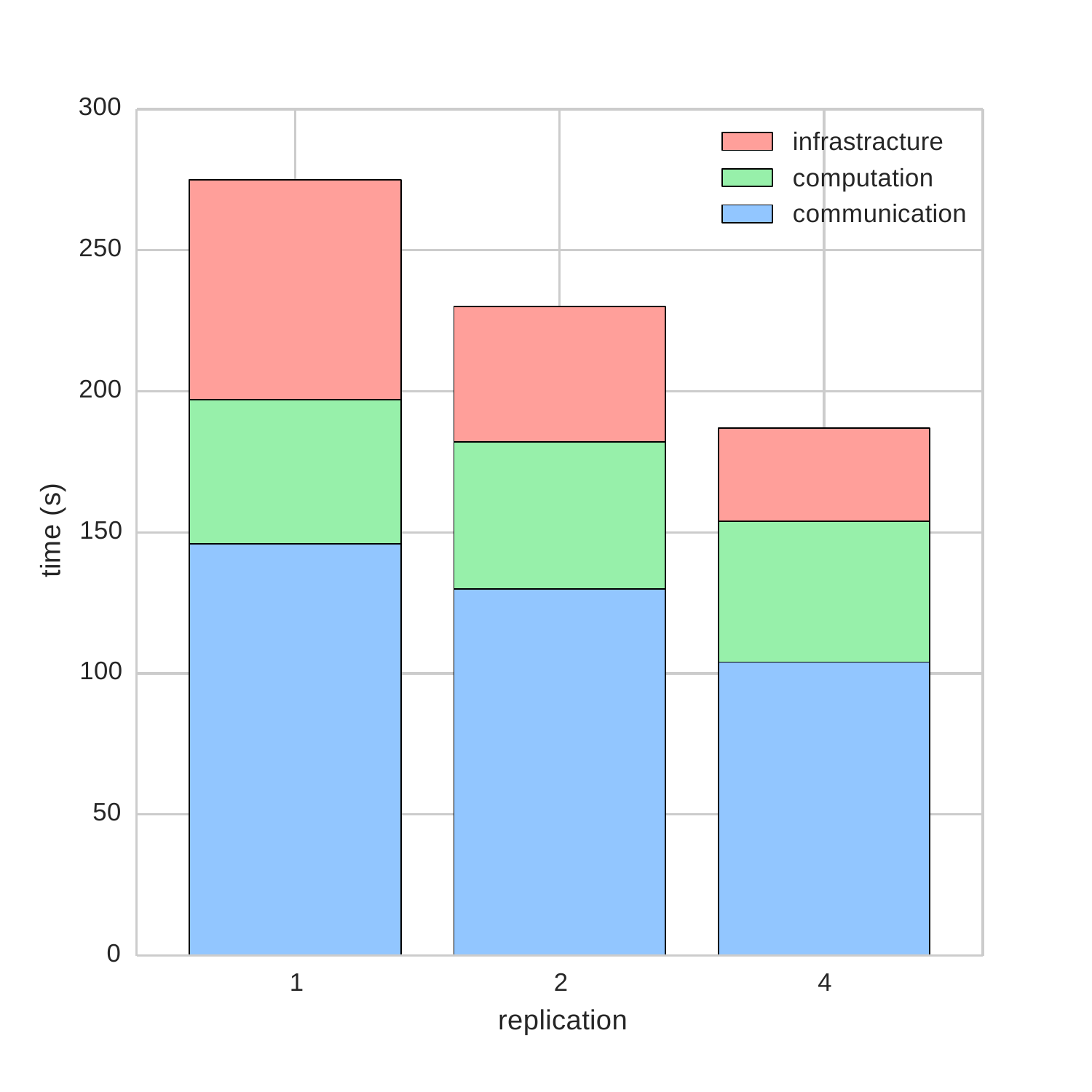}
    \caption{ $\sn=16000$, $\rho=\{1,2,4\}$.}
    \label{fig:components-plot-time-16000.pdf}
  \end{subfigure}
  \hfil
  \begin{subfigure}[b]{\columnwidth}
    \includegraphics[width=\textwidth]{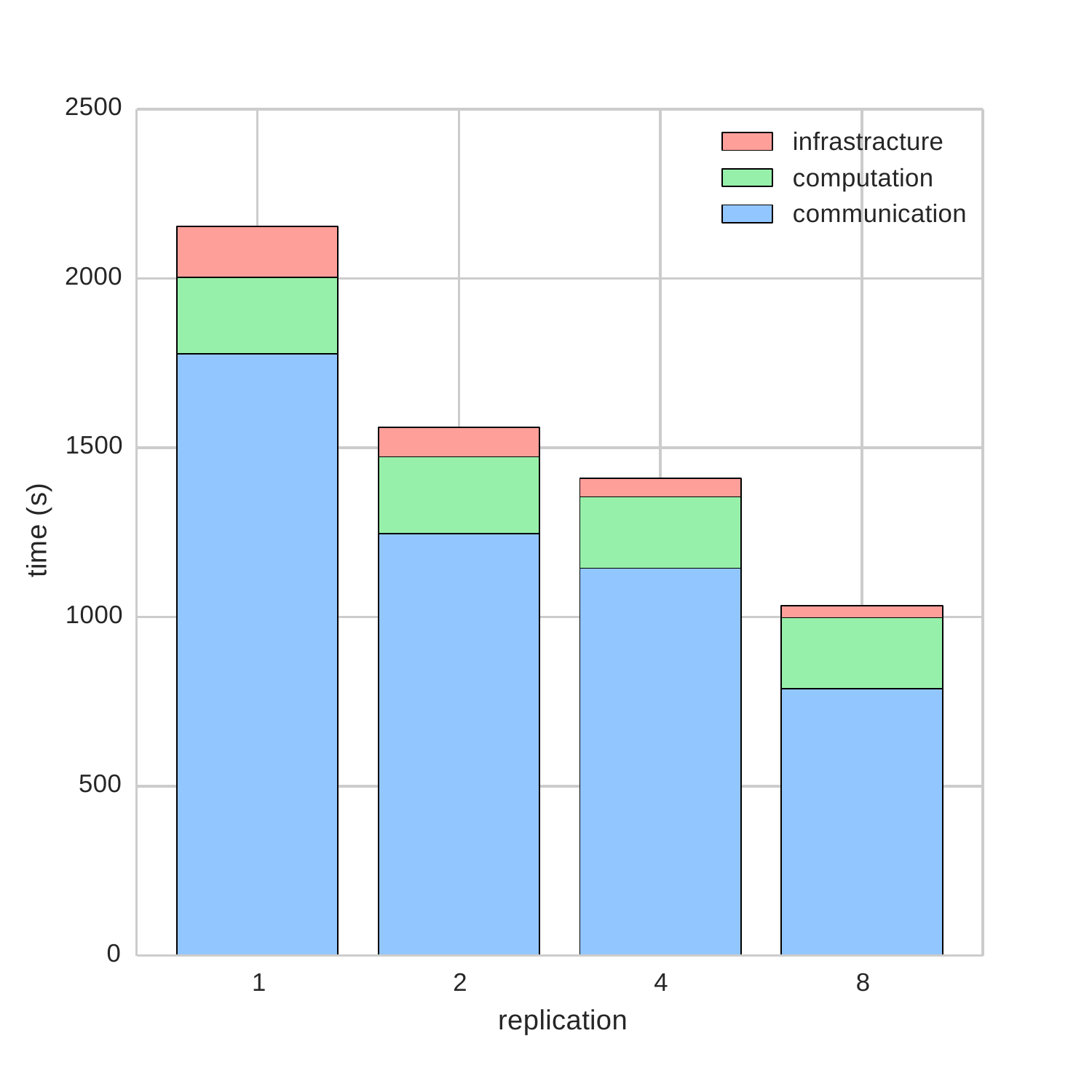}
    \caption{ $\sn=32000$, $\rho=\{1,2,4,8\}$.}
    \label{fig:components-plot-time-32000.pdf}
  \end{subfigure}
  
  \caption{Component cost vs replication. Each bar shows the time of
    the communication, computation, and infrastructure components.}
\end{figure*}

We are ready to study the performance of a multi-round approach. We
recall that, for any replication factor $1\leq \rho \leq \snm$, the
number of rounds is $\sn/(\rho \sm)+1$.
Figures~\ref{fig:replication-plot-16000}
and~\ref{fig:replication-plot-32000} investigate the running times
with different values of the replication factor and input size
$\sn=16000$ and $\sn=32000$, respectively.  For $\rho=\snm$, we get a
monolithic approach (i.e., two rounds).  Without any surprise, the
best running time is reached by the monolithic approach: this supports the
common practice to minimize round number, but only if the total amount of
communication remains unchanged.  However, we
observe that a multi-round approach increases the running time with
respect to the monolithic two-round approach by an average factor
$7\%$ for each additional round.  We believe that the main reason of
this increase is due to the distributed file system HDFS, which is used for
storing pairs between rounds and is optimized for writing and reading large
files.  When $\rho=\snm$, each
reduce task writes two large chunks of output pairs (one per round) on
 HDFS containing all outputs of the associated reducers in a given
round.  On the other hand, for values of $\rho$ smaller than $\snm$,
each reduce task writes a larger number of chunks of smaller
size. Although the total size of the written files is the same in each
approach, the system is not able to exploit the features of HDFS with
a multi-round approach.  We conjecture that the performance gap can be closed by
exploiting other implementations that use local file systems for
storing intermediate pairs (e.g., Spark).

Figures~\ref{fig:replication-plot-16000}
and~\ref{fig:replication-plot-32000} also show the running time of
each round: for each histogram bar, the $i$-th colored block from the
bottom denotes the time of the $i$-th round. The work is evenly
distributed among rounds and the cost of each round decreases with the
replication factor $\rho$. In each run, the last round is faster than
the remaining ones since the reduce function only performs the
addition of $\rho$ submatrices.
%
Finally, we observe that the algorithm efficiently scales with the
input size:  the running time, for any replication factor, increases
on the in-house cluster by a factor $\sim 8$ when the input side
doubles, which agrees with the cubic (in the matrix side) complexity of
the algorithm.

\paragraph{Q3: Cost analysis.}
We now analyze how communication, computation and infrastructure
affect the running time. We define the following three costs:

\begin{itemize}[itemsep=-1ex,topsep=-2ex]
\item \emph{Infrastructure cost $T_{infr}$:} It is the time for setting up
  the Hadoop system and each round. It is given by the time of the
  algorithm when no computation is done in the reduce functions and
  the value of each pair is replaced by an empty matrix (i.e., the
  amount of shuffled data is negligible).

\item \emph{Computation cost $T_{comp}$:} It is the time for performing the
  local computation. Let $\hat T_{comp}$ be the time of the algorithm
  when any local multiplication is performed on two locally created
  matrices (the cost of the creation is negligible) and the value of
  each pair is replaced by an empty matrix. Then, the computation cost
  is $\hat T_{comp}-T_{infr}$.

\item \emph{Communication cost $T_{comm}$:} It is the time for exchanging
  pairs between mappers and reducers, and includes the costs of
  reading/writing data on HDFS and of shuffle steps. Let $\hat
  T_{comm}$ be the time of the algorithm when no computation is done
  in the reduce function (the output pair is a copy of an input
  pair). Then, the communication cost is $\hat T_{comm}-T_{infr}$.\\
\end{itemize}

We think that the infrastructure cost is a good approximation of the
true time for setting up the system and each round. However, this is
not true for the communication and computation costs since the
procedure does not take into account the concurrency between the three
components.  Nevertheless, these costs still provide an interesting
overview of the main factors determining the running time of an
algorithm. (To the best of our knowledge no tools for cost analysis are
available.)
Figures~\ref{fig:components-plot-time-16000.pdf}
and~\ref{fig:components-plot-time-32000.pdf} show the costs of the
three components for $\sn = 16000$ and $\rho\in\{1,2,4\}$, and for
$\sn=32000$ and $\rho\in\{1,2,4,8\}$, respectively.  The
infrastructure cost increases linearly with the round number, and the
average fixed cost of a round is $17$ seconds.  The computation cost
for a given input size is independent of the replication value, which
confirms that the work is evenly distributed among cluster nodes.  The
communication cost increases when $\rho$ increases as already noted for
Figures~\ref{fig:replication-plot-16000} and~\ref{fig:replication-plot-32000},
and dominates the total time.



\begin{figure}[t]
  \centering
  \includegraphics[width=\columnwidth]{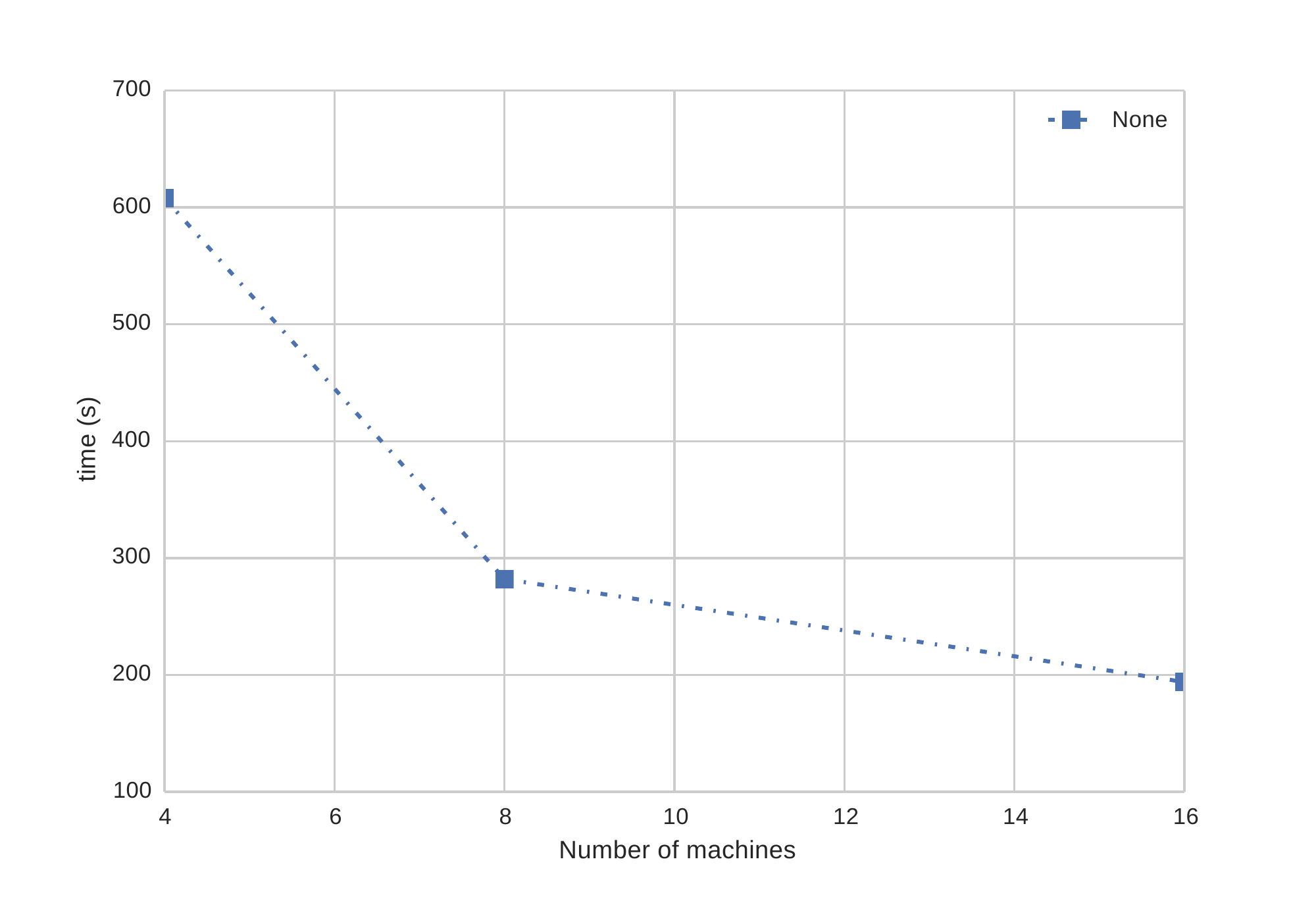}
  \caption{Time vs number of nodes with $\sn=16000$. The experiments
    have been carried out with replication $\rho\in\{1,2,4\}$ and
    $p\in\{4,8,16\}$ nodes.}
  \label{fig:scalability}
\end{figure}

\paragraph{Q4: Scalability.}
We now study the scalability of the algorithm by running the dense
algorithm on a smaller number of nodes.  For this experiment, we run
the dense algorithm with input size 16000 on 4, 8 and 16 nodes of the
in-house cluster. The results are given in
Figure~\ref{fig:scalability} with different replication factors
$\rho\in\{1,2,4\}$.  The algorithm scales efficiently
although there is a small reduction in the speed up with 16 nodes.  This may be
due to a loss in data locality and to a larger cost of the shuffle step when the
algorithm is run on 16 nodes.

\paragraph{Q5: 2D vs. 3D algorithms}

\begin{figure}[t]
  \centering
  \includegraphics[width=\columnwidth]{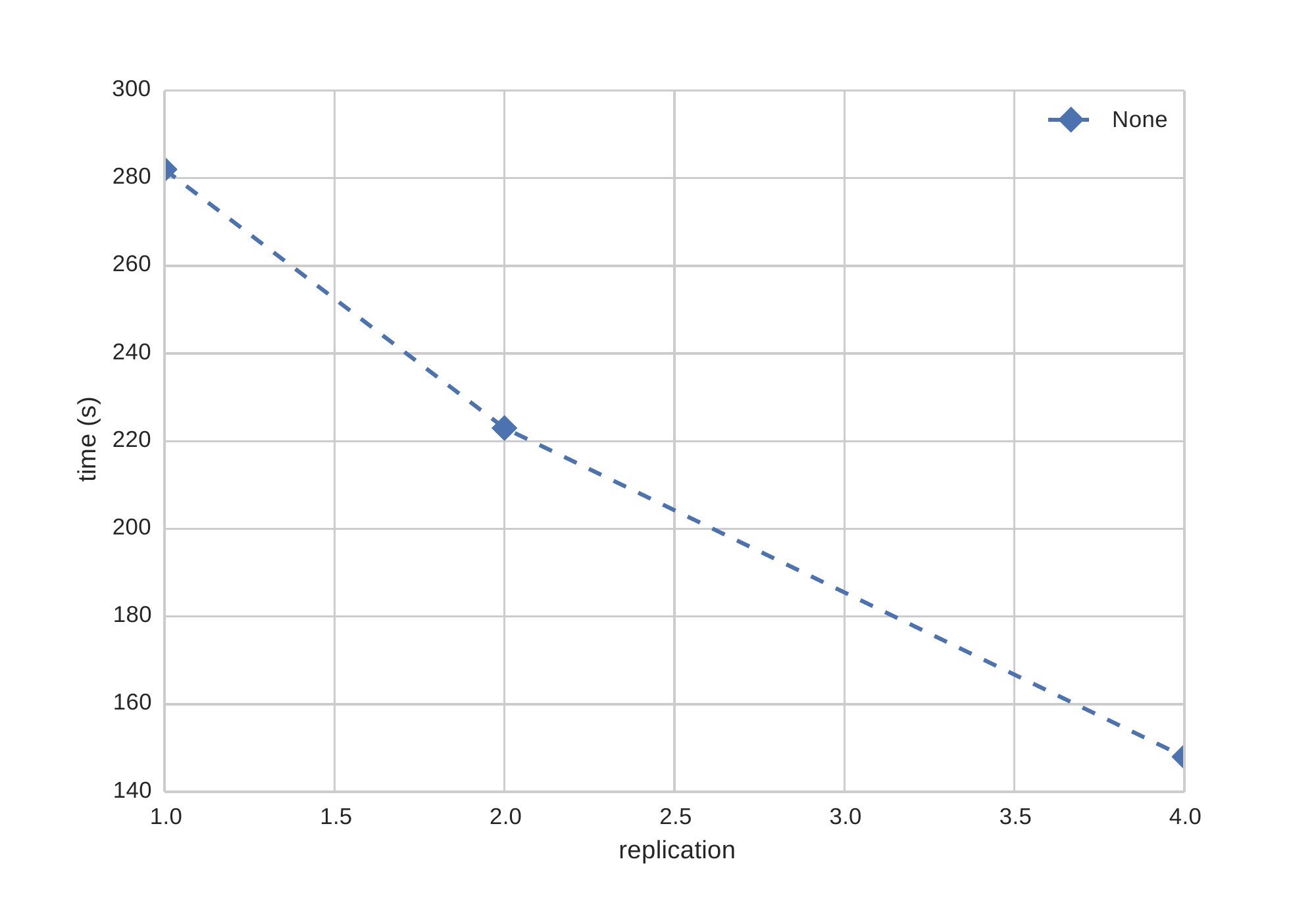}
  \caption{Comparison between the 2D and 3D approaches
    with $\sn=16000$. The replication factor is $\rho\in\{1,2,4\}$
    for the 3D approach, and $\rho\in\{1,2,4,8,16\}$ for the 2D
    approach.}
  \label{fig:comparison-2d.pdf}
\end{figure}

We now move to analyze the performance of the 3D multiplication
strategy compared with the 2D
strategy. Figure~\ref{fig:comparison-2d.pdf} clearly shows that the 3D
approach has a significant performance advantage. This is due to the
fact that the total shuffle size is, for the 3D approach,
$\BO{n\sqrt{n/m}}$, whereas for the 2D approach it is
$\BO{n^2/m}$. Since the major bottleneck in MapReduce is
communication, the 2D approach incurs a significant penalty.

\paragraph{Q6: Sparse matrices.}

\begin{figure}[t]
  \centering
  \includegraphics[width=\columnwidth]{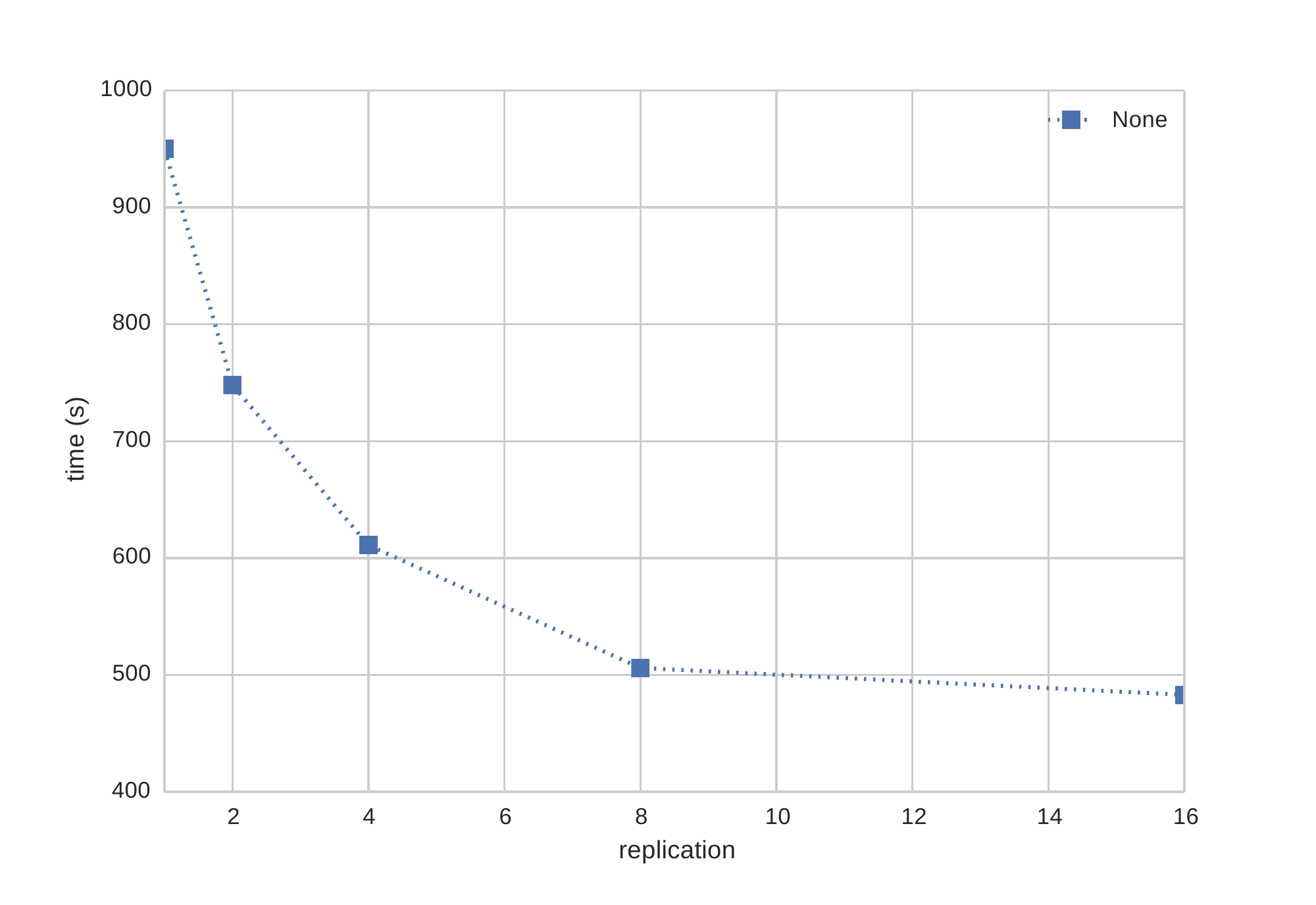}
  \caption{Time vs replication for sparse matrix
    multiplication with $\sn\in\{2^{20}, 2^{22}, 2^{24}\}$ and an
    average of 8 non-zero entries per row and per column.  For each
    input size, we consider all possible replication factors, from 0
    to $\sn/\smp$ ($\smp$ is set to $\{2^{18}, 2^{19}, 2^{20}\}$,
    respectively).}
  \label{fig:sparse-big-blocks.pdf}
\end{figure}

We now investigate the performance of the 3D sparse
algorithm and show that it improves performance by exploiting input
sparsity.  Figure~\ref{fig:sparse-big-blocks.pdf} shows the running
times with different replication factors of the 3D sparse algorithm
with two input Erd\"{o}s-R\'enyi matrices with $\sn\in \{2^{20},
2^{22}, 2^{24}\}$, and an average of 8 non-zero entries per row and
per column (i.e., $\delta\in\{1/2^{17}, 1/2^{19}, 1/2^{21}\}$,
respectively). Since the output matrices have expected density
$\delta_O\in\{1/2^{14}, 1/2^{16}, 1/2^{18}\}$ (see
Section~\ref{sec:mmult}), the input matrices are partitioned into
submatrices of size $2^{18}$, $2^{19}$ and $2^{20}$, respectively, so
that the expected number of non-zero elements in each submatrix of the
output is comparable with the subproblem size of the dense case. This
shows that, by exploiting the sparseness of the input matrix, we can
tackle much bigger problems, under the same memory constraints.  As
mentioned in Section~\ref{sec:implementation}, we avoided the
computation of local products because of the lack of an efficient
Java implementation of sequential sparse matrix multiplication. However, as we
have experimentally
shown, the time spent computing local products is a fixed additive
factor, whereas the communication is the dominant component, thus
these experiments clearly show the tradeoffs between round number 
and time for sparse matrix multiplication.

\subsection{Experiments on Amazon Elastic MapReduce}
\label{sec:experiments-amazon}

In this section we report our investigations in Amazon Elastic
MapReduce (EMR) targeting the questions Q1, Q2 and Q3 described at the
beginning of Section~\ref{sec:experiments}.

\paragraph{Q1: Impact of subproblem size $m$.}
These experiments show a similar behavior on c3.8xlarge and i2.xlarge
instances, with $\sm=4000$ being the optimal choice.  Interestingly,
smaller instance types (like c3.xlarge) require $\sm=2000$ for
avoiding memory errors.  All the following experiments assume
$\sm=4000$.

\paragraph{Q2: Impact of replication factor $\rho$.}
The experiments have been carried out on EMR with c3.8xlarge instances
and the results are in Figures~\ref{fig:amazon-replication-16000.pdf}
and~\ref{fig:amazon-replication-32000.pdf} with
$\sn\in\{16000,32000\}$, respectively.  Running times significantly
increase with respect to experiments on the in-house cluster: the
running times with $\sn=16000$ (resp., $\sn=32000$) on EMR are on the
average $4.7$ (resp., $1.4$) times larger than the ones on the
in-house cluster, even if the computational resources of the two
systems are somewhat similar.  It is interesting to note that the gap
decreases with larger input sizes.  The average performance loss with
respect to the monolithic two-round approach is $17\%$ for each
additional round.

As for the scalability, we observe that the scaling factor on EMR is
$\sim 5$, smaller than the one achieved on the in-house
cluster ($\sim 8$). This is due to high fixed costs which are not
efficiently amortized with small inputs.

\paragraph{Q3: Cost analysis.}
The cost analysis on EMR with c3.8xlarge instances is reported in
Figures~\ref{fig:amazon-components-time-16000.pdf}
and~\ref{fig:amazon-components-time-32000.pdf} for $\sn = 16000$ and
$\rho\in\{1,2,4\}$, and for $\sn=32000$ and $\rho\in\{1,2,4,8\}$,
respectively.  It should be noted that the average computation cost is
similar to the respective component in the in-house cluster, although
there is a larger variance due to the unpredictable load of the
physical machines of EMR. The average infrastructure cost is 30
seconds.

Finally, we report in
Figure~\ref{fig:amazon-i3-components-time-16000.pdf} the cost analysis
with $\sn=16000$ and $\rho\in\{1,2,4\}$ on i2.xlarge instances. A
i2.xlarge instance has faster disks optimized for random I/Os but
slower network than a c3.8xlarge instance.  However, we observe that
the communication costs are smaller than the respective costs in
Figure~\ref{fig:amazon-components-time-16000.pdf} for c3.8xlarge
instances.  This fact supports the claim in
Section~\ref{sec:experiments-cluster} (question {\bf Q2}) that the
main bottleneck is the inability of HDFS to efficiently read/write
small chunks of data.

\begin{figure}[t]
  \centering
  \includegraphics[width=\columnwidth]{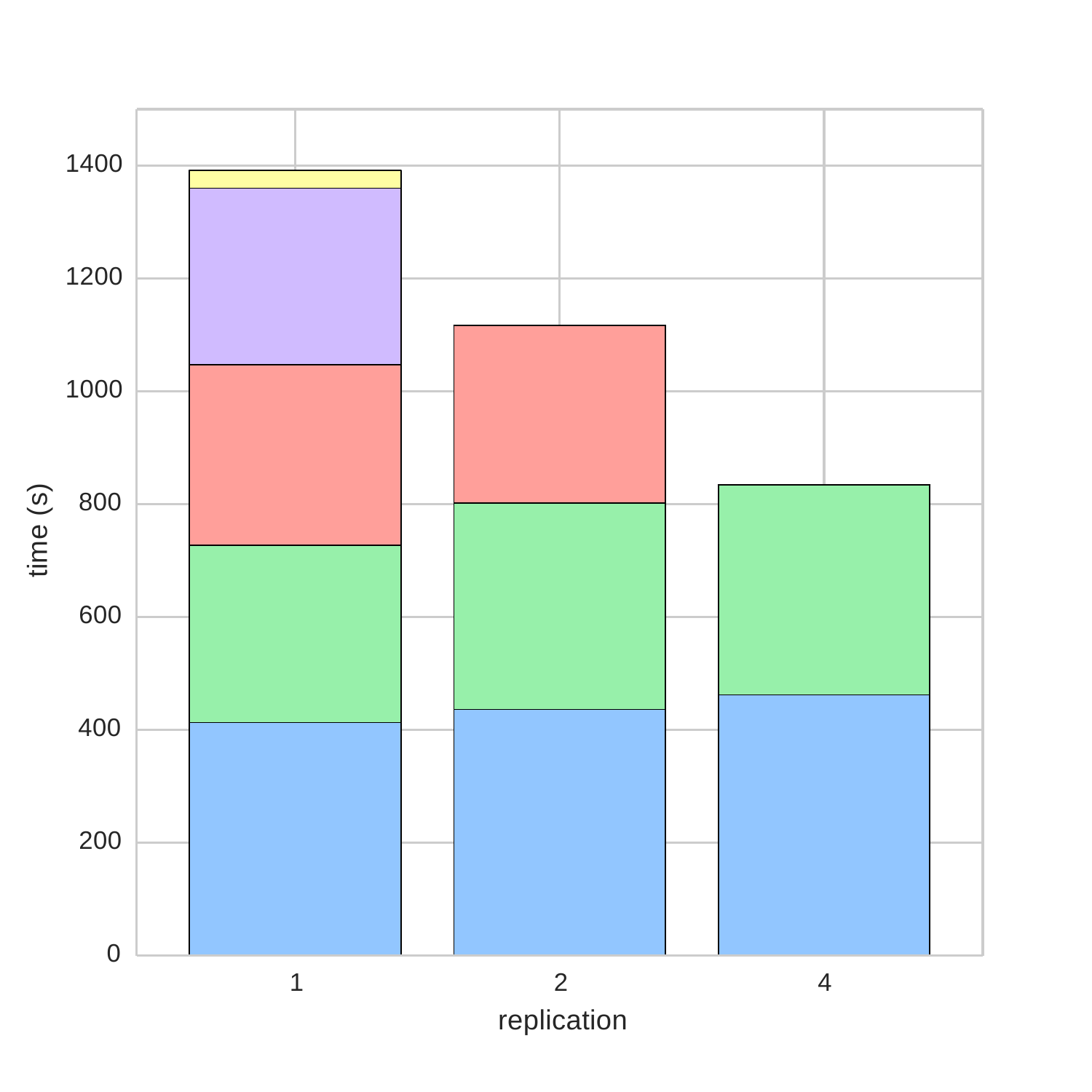}
  \caption{Time vs replication with $\sn=16000$ and $\rho=\{1,2,4\}$
    on c3.8xlarge instances. In each bar of the histogram, the $i$-th
    colored block denotes the time of the $i$-th round.}
  \label{fig:amazon-replication-16000.pdf}
\end{figure}

\begin{figure*}
  \centering
  \begin{subfigure}[b]{\columnwidth}
    \includegraphics[width=\textwidth]{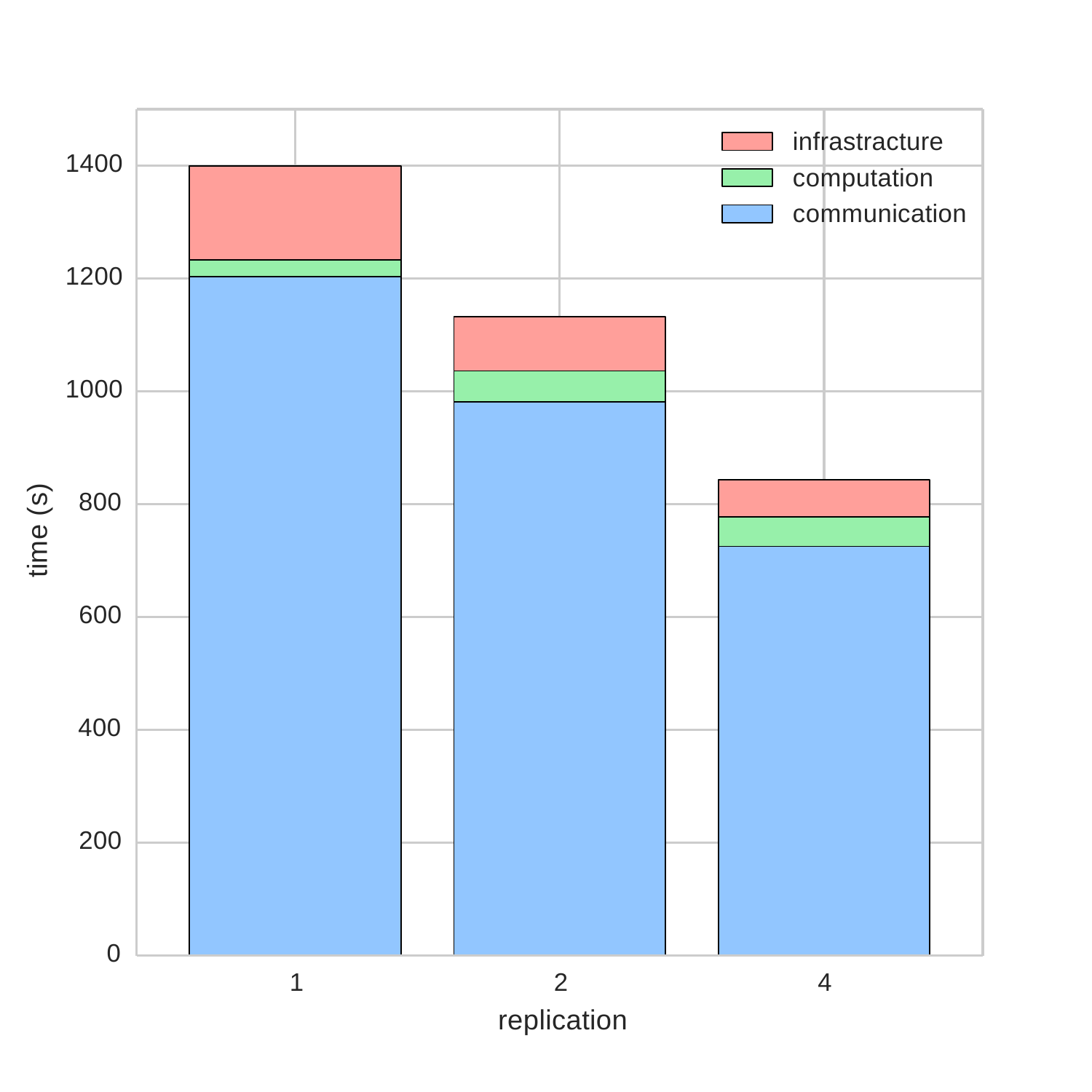}
    \caption{ Component cost vs replication with $\sn=16000$,
      $\rho=\{1,2,4\}$ on c3.8xlarge instances. Each bar shows the
      time of the communication, computation, and infrastructure
      components.}
    \label{fig:amazon-components-time-16000.pdf}
  \end{subfigure}
  \hfil
  \begin{subfigure}[b]{\columnwidth}
    \includegraphics[width=\textwidth]{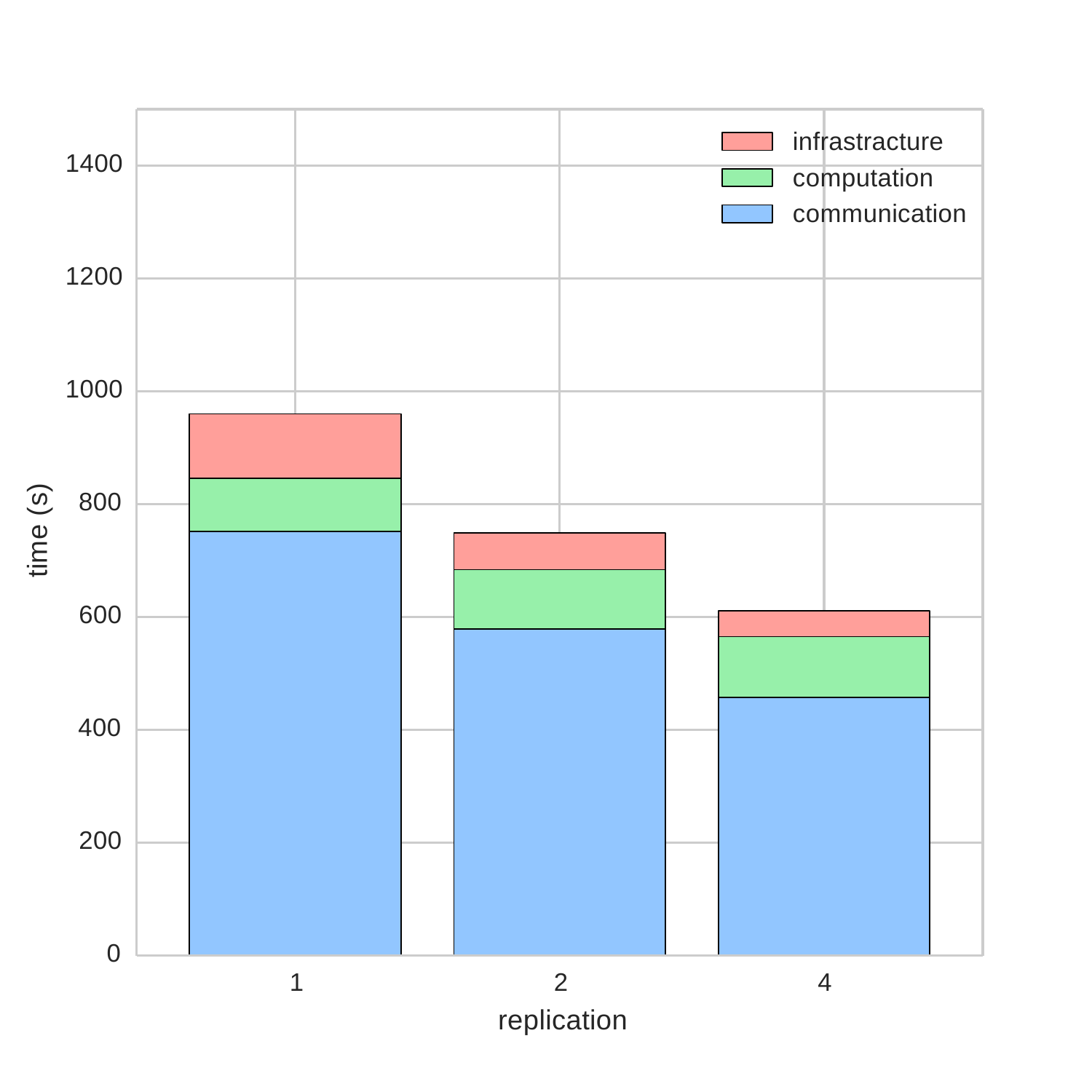}
    \caption{ Component cost vs replication with $\sn=16000$,
      $\rho=\{1,2,4\}$ on i2.xlarge instances. Each bar shows the time of the
      communication, computation, and infrastructure components.}
    \label{fig:amazon-i3-components-time-16000.pdf}
  \end{subfigure}

  \caption{Component cost for multiplying $16000\times 16000$ matrices
    on Amazon EMR, using different instance types. Instances of type
    i2.xlarge (Figure~\ref{fig:amazon-i3-components-time-16000.pdf})
    have a faster disk with respect to c3.8xlarge instances
    (Figure~\ref{fig:amazon-components-time-16000.pdf}).}
\end{figure*}

\begin{figure*}
  \centering

  \begin{subfigure}[b]{\columnwidth}
    \includegraphics[width=\textwidth]{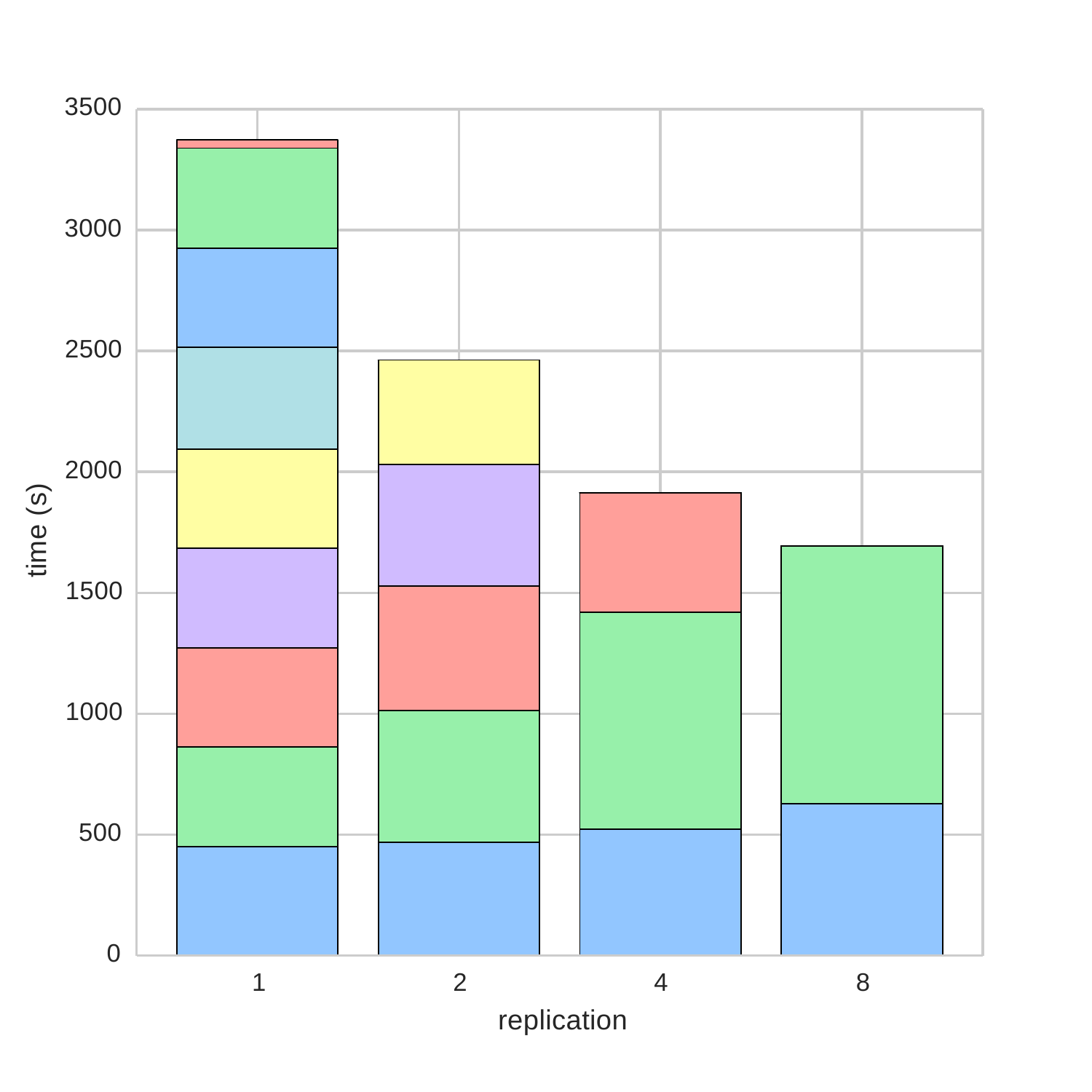}
    \caption{ Time vs replication with $\sn=32000$, $\rho=\{1,2,4,8\}$
      on c3.8xlarge instances. In each bar of the histogram, the
      $i$-th colored block denotes the time of the $i$-th round.}
    \label{fig:amazon-replication-32000.pdf}
  \end{subfigure}
  \hfil
  \begin{subfigure}[b]{\columnwidth}
    \includegraphics[width=\textwidth]{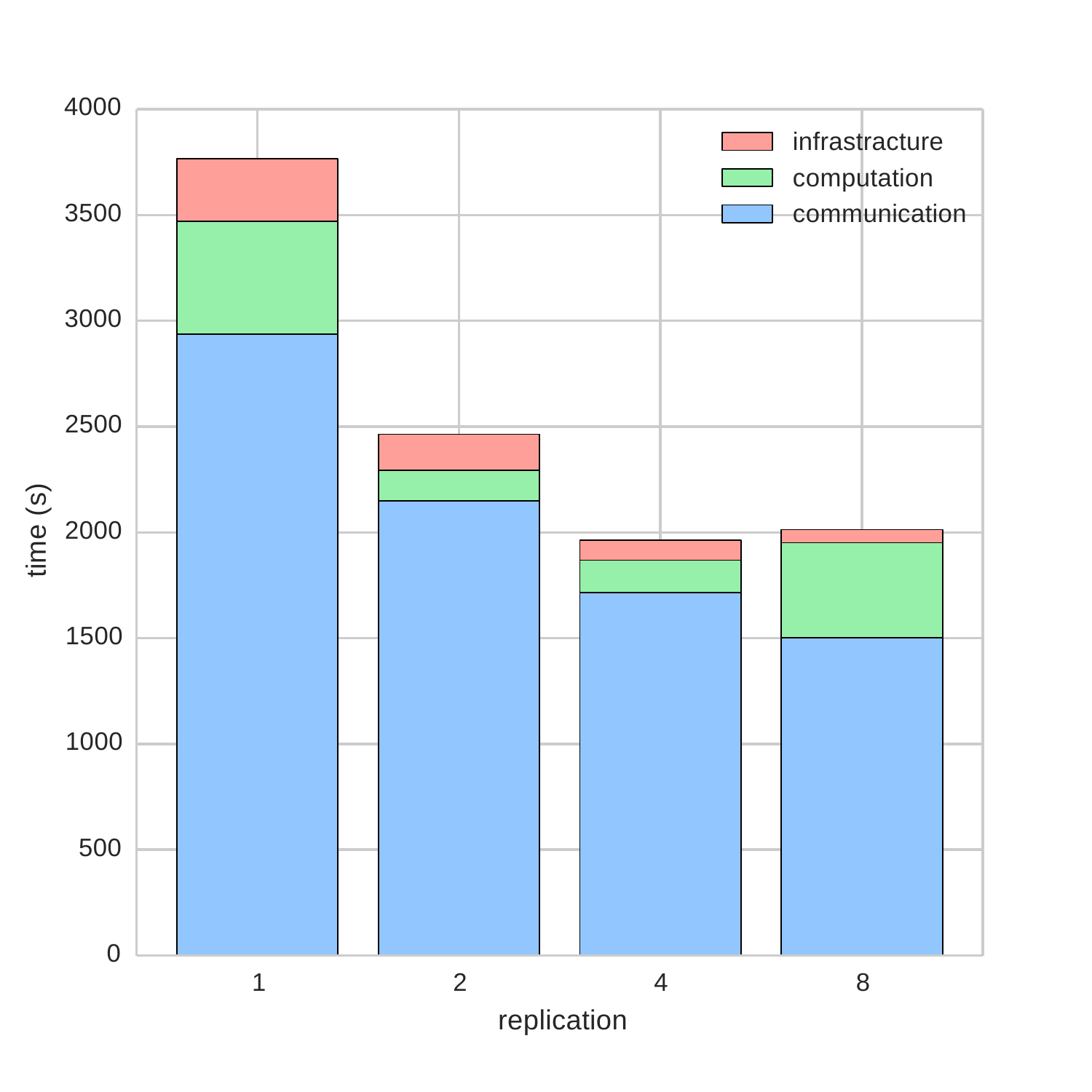}
    \caption{ Component cost vs replication with $\sn=32000$ and
      $\rho=\{1,2,4,8\}$ on c3.8xlarge instances. Each bar shows the
      time of the communication, computation, and infrastructure
      components.}
    \label{fig:amazon-components-time-32000.pdf}
  \end{subfigure}

  \caption{Experiments on Amazon EMR for multiplying $32000\times
    32000$ dense matrices.}

\end{figure*}


\section{Conclusion}\label{sec:conclusion}
In this paper we have proposed the Hadoop library \m for performing
dense and sparse matrix multiplication in MapReduce by exploiting the
theoretical results in~\cite{PietracaprinaPRSU12}, and we have carried
out an extensive experimental study on an in-house cluster and Amazon
Web Services.  The results give evidence that multi-round algorithms
can have performance comparable with monolithic ones (assuming a
similar total amount of communication) even on the Hadoop framework,
which is not the most suitable MapReduce implementation for executing
multi-round algorithms.  Moreover, experiments suggest that the common
attitude to only focus on round number when designing MapReduce
algorithms could significantly reduce performance if it implies a
larger amount of total communication. This fact supports recent
computational models for MapReduce that mainly focus on the
minimization of the total communication
complexity~\cite{GoodrichSZ11}, or that aim at reducing round number
when an upper bound on the allowed total communication complexity is
given~\cite{PietracaprinaPRSU12}.

As future work, we plan to further investigate sparse matrix
multiplication on MapReduce and to test our algorithms on other
implementations of the MapReduce paradigm. In particular, we are
currently developing our algorithms in the Spark framework, where the
management of the input/output pairs of each round is more efficient
than Hadoop.


\section*{Acknowledgments}
This paper was supported in part by: MIUR of Italy under project
AMANDA; University of Padova under projects CPDA121378 and AACSE;
Amazon in Education Grant; equipment donation by Samsung. The authors
are grateful to Paolo Rodeghiero for initial discussions on Hadoop and
to Andrea Pietracaprina and Geppino Pucci for useful insights.


\end{document}